\documentclass[nofootinbib,aps,11pt,preprintnumbers,unsortedaddress,prd]{revtex4-2}
\usepackage[utf8]{inputenc}
\usepackage{amsmath,amssymb}
\usepackage{bbold}
\usepackage{yhmath}
\usepackage{subcaption}
\usepackage{graphicx}
\usepackage{hyperref}
\usepackage{mathtools}

\newcommand{\rI}{{\mathrm{I}}}
\newcommand{\rII}{{\mathrm{II}}}
\newcommand{\rIII}{{\mathrm{III}}}

\newcommand{\mH}{\ensuremath{\mathcal{H}}}
\newcommand{\mHc}{\ensuremath{\mathcal{H}_c}}

\newcommand{\td}{\ensuremath{\,\text{d}}}
\DeclareMathOperator{\WF}{WF}
\DeclareMathOperator{\supp}{supp}
\newcommand{\bR}{\ensuremath{\mathbb{R}}}
\newcommand{\bC}{\ensuremath{\mathbb{C}}}
\newcommand{\bD}{\ensuremath{\mathbb{D}}}
\newcommand{\bN}{\ensuremath{\mathbb{N}}}
\newcommand{\bP}{\ensuremath{\mathbb{P}}}
\newcommand{\bS}{\ensuremath{\mathbb{S}}}
\newcommand{\mC}{\ensuremath{\mathcal{C}}}
\newcommand{\mD}{\ensuremath{\mathcal{D}}}

\newcommand{\mN}{\ensuremath{\mathcal{N}}}

\newcommand{\norm}[1]{\ensuremath{\left\Vert #1 \right\Vert}}

\usepackage{amsthm}
\newtheorem{thm}{Theorem}[section]
\newtheorem{prop}[thm]{Proposition} 
\newtheorem{lem}[thm]{Lemma}

\theoremstyle{definition}
\newtheorem{defi}{Definition}[section]

\begin{document}

\title{Construction of the Unruh State for a Real Scalar Field on the Kerr-de Sitter Spacetime}
\author{Christiane K.M. Klein}
\email{klein@itp.uni-leipzig.de}
\affiliation{Institut f\"ur Theoretische Physik, Universit\"at Leipzig,\\ Br\"uderstra{\ss}e 16, 04103 Leipzig, Germany}

\begin{abstract}
The study of physical  effects of quatum fields in black hole spacetimes, which is related to questions such as the validity of the strong cosmic censorship conjecture, requires a Hadamard state describing the physical situation.
Here, we consider the theory of a free scalar field on a Kerr-de Sitter spacetime, focussing on spacetimes with sufficiently small angular momentum of the black hole and sufficiently small cosmological constant. We demonstrate that an extension of the Unruh state, which describes the expected late-time behaviour in spherically symmetric gravitational collapse, can be rigorously constructed for the free scalar field on such Kerr-de Sitter spacetimes. In addition, we show that this extension of the Unruh state is a Hadamard state in the black hole exterior and in the black hole interior up to the inner horizon. This provides a physically motivated Hadamard state for the study of free scalar fields in rotating black hole spacetimes.
\end{abstract}

\maketitle

\section{Introduction}\label{sec1}

Recently, there has been a renewed interest in the behaviour of quantum fields in black hole spacetimes of charged or rotating black holes \cite{Ottewill:2000, Levi:2016, Lanir:2017, Sela:2018, Hollands:2019, Hollands:2020, Zilberman:2019, Klein:2021, Zilberman:2021, Zilberman:2022a, Zilberman:2022b}. The behaviour of the field near the inner horizon is particularly interesting, because it is connected to the strong cosmic censorship conjecture \cite{Penrose:1974, Christodoulou:2008}, which holds in the linear regime classically for Kerr-de Sitter \cite{Dias:2018}, but is violated in Reissner-Nordström-de Sitter \cite{Cardoso:2017, Cardoso:2018, Dias:2018a}.

An important open question in theoretical physics today is the merger of general relativity and quantum field theory into a theory of quantum gravity. While no such theory is known as of yet, one step towards it is restricting its possible low-energy behaviour by studying quantum field theory on curved spacetimes. 

One feature of quantum field theory in curved spacetimes is that even for a free scalar field, there is no unique ground state on a generic curved spacetime (see \cite{Fewster:2019} for the definition of a ground state in the algebraic framework). Thus, there is also no preferred Fock space build from (finite) excitations of such a state. More generally, it is no longer clear which of the unitarily inequivalent Hilbert space representations one should choose for the quantum theory \cite{Wald:1995}.

When studying a physical effect of some quantum field in a curved spacetime, an important first step is the identification of a quantum state or a class of quantum states which adequately describes the given physical situation. This implies that the state should satisfy the Hadamard property. This property is a regularity requirement which is necessary to allow for the assignment of finite expectation values  with finite fluctuations to non-linear observables \cite{Hollands:2001,Hollands:2001b}, for example the stress-energy tensor of the quantum field. 

For scalar quantum fields on Schwarzschild spacetimes, there are two well-studied options: the Hartle-Hawking state \cite{Hartle:1976,Israel:1976,Sanders:2013,Gerard:2021} and the Unruh state \cite{Unruh:1976, Dappiaggi:2009}. The Hartle-Hawking state is a global thermal equilibrium state at the Hawking temperature of the black hole, which corresponds to the black hole's surface gravity divided by $2\pi$.
 
In contrast, the Unruh state is not a thermal equilibrium state. It is a stationary state that can be thought of as describing a hot body, namely the black hole, immersed in vacuum. In particular, it contains no particles coming from past null infinity, while at future null infinity, one finds an energy flux consistent with black-body radiation at the Hawking temperature. Due to this, it is generally considered to be the appropriate state for the description of spherically symmetric gravitational collapse \cite{Candelas:1980, Balbinot:1984, Balbinot:2000, Dappiaggi:2009}.

Physically, it is clear how to construct the Unruh state \cite{Unruh:1976} and its analogues on other black hole spacetimes. However, the rigorous proof of their existence and Hadamard property are quite difficult.

So far, analogues of the Unruh state have been costructed, includig a proof of their Hadamard property, on Schwarzschild de-Sitter \cite{Brum:2014} and Reissner-Nordström-de Sitter \cite{Hollands:2019} spacetimes, as well as for massless fermions on slowly rotating Kerr spacetimes \cite{Gerard:2020}. But to our knowledge, an analogue of the Unruh state on Kerr or Kerr-de Sitter spacetimes for scalar fields has not been rigorously constructed as of yet.

One of the main difficulties in extending the previous results for the scalar field to spacetimes with rotating black holes is that, due to the appearance of the ergosphere, the exterior region of the Kerr(-de Sitter) spacetime is not static. The static nature of the black hole exterior region is necessary for the proof of the Hadamard property of the Unruh state in the form given in \cite{Dappiaggi:2009, Brum:2014} for the Schwarzschild (-de Sitter) spacetime. Hence, a direct adaptation of this proof is not possible.  

In this paper, we will define the Unruh state on Kerr-de Sitter spacetimes with sufficiently slow rotation and sufficiently small cosmological constant and prove its Hadamard property. We will combine the techniques used in \cite{Dappiaggi:2009,Brum:2014, Hollands:2019}, with ideas developed in \cite{Gerard:2020}, which we generalize to the Kerr-de Sitter spacetime. These ideas enable us to prove the Hadamard condition in some subregion of the black hole exterior. A careful analysis of some arguments from \cite{Hollands:2019} then allow us to extend the proof to the whole spacetime.

The rest of the paper is organised as follows. In section \ref{sec:GemSet} we introduce the geometric setup of the spacetime.  Section \ref{sec:scalar field} introduces the scalar field. The Unruh state is defined in section \ref{sec:Unruh} and its Hadamard property is shown in section \ref{sec:Had}. We briefly summarize in section \ref{sec:sum}.
Throughout the paper we work in geometrical units $\hbar=c=G=k_B=1$.

\section{Geometric setup}
\label{sec:GemSet}
In this paper, we are considering an axisymmetric, rotating, non-charged black hole in the presence of a positive cosmological constant $\Lambda$. The cosmological constant $\Lambda$, as well as the black hole mass M and the angular momentum parameter $a$ should be chosen in such a way that the function
\begin{align}
\label{eq:Delta_r}
 \Delta_r&=(1-\lambda r^2)(r^2+a^2)-2Mr\, ,
\end{align}
$\lambda=\Lambda/3$, has three distinct real, positive roots $r_-<r_+<r_c$. In particular, we set $M=1$. The admissible parameter range in the $(a,\lambda)$-plane is depicted in Figure \ref{fig:parRegs}. Here, we consider $0<\lambda<1/27$ and $0<a$ sufficiently small. In this case, the Boyer-Lindquist blocks  $\rI=\bR_t\times(r_+,r_c)\times (\bS^2_{\theta,\varphi})$, $\rII=\bR_t\times(r_-,r_+)\times (\bS^2_{\theta,\varphi})$, and  $\rIII=\bR_t\times(r_c,\infty)\times (\bS^2_{\theta,\varphi})$ are all non-empty. The regions $\rI$ and $\rIII$ are the exterior of the black hole, with $\rIII$ being the region beyond the cosmological horizon. The region $\rII$ is the interior of the black hole up to its inner horizon.

The metric on these blocks in Boyer-Lindquist coordinates $(t,r,\theta,\varphi)$ is given by \footnote{This coordinate system does not cover the axis where $\sin \theta =0$. However, it can be shown that this metric can be extended to the axis \cite{Borthwick:2018}.}
\begin{align}
    g&=\frac{\Delta_\theta a^2\sin^2\theta-\Delta_r}{\rho^2\chi^2}\td t^2+\left[\Delta_\theta(r^2+a^2)^2-\Delta_ra^2\sin^2\theta\right]\frac{\sin^2\theta}{\rho^2\chi^2}\td\varphi^2 \\
\nonumber &+2\frac{a\sin^2\theta}{\rho^2\chi^2}[\Delta_r-\Delta_\theta(r^2+a^2)]\td t\td\varphi+\frac{\rho^2}{\Delta_r}\td r^2+\frac{\rho^2}{\Delta_\theta}\td \theta^2\, ,
\end{align}
where
\begin{align}	
\label{eq:D_t,rho,chi}
     \Delta_\theta &=  1 + a^2 \lambda \cos^2\theta &  \rho^2&=r^2+a^2\cos^2\theta &  \chi&=1+a^2\lambda\, .
\end{align}
We will chose $\partial_t$ to be future-pointing in the part of $\rI$ where it is timelike.

In order to join multiple of these blocks, one can introduce so-called $KdS*$- and $*KdS$-coordinates, which allow a continuation of the Boyer-Lindquist blocks through the ingoing or outgoing piece of the horizon $\{r=r_i\}$, $i\in\{-,+,c\}$ \cite{Borthwick:2018}. The spacetime we consider will then be the block $\rI$ joined in the $KdS*$-coordinates via $\mH^R\subset\{r=r_+\}$ to block $\rII$ and in the $*KdS$-coordinates via $\mHc^L \subset\{r=r_c\}$ to block $\rIII$. We will refer to this spacetime as $M$.  

The coordinates used for most of the computations in this paper are a combination of the $KdS*$- and $*KdS$-coordinates, tailored to one of the horizons $\{r=r_X\}$. Taking $X \in\{+,c\}$, they are defined by
\begin{align}
\td v_X&= \td t+\frac{\chi(r^2+a^2)}{\Delta_r}\td r\, , & \td u_X&= \td t-\frac{\chi(r^2+a^2)}{\Delta_r}\td r\, ,\\\nonumber
\td \theta_X&=\td\theta\, , & \td \varphi_X&=\td \varphi-\frac{a}{r_X^2+a^2}\td t\, . 
\end{align}
We will also call $r_*(r)$, defined by $\td r_* = \chi (r^2+a^2)/\Delta_r \td r$, the "tortoise coordinate".
The coordinates $u_X$ and $v_X$ range from $-\infty$ to $\infty$ in each of the Boyer-Lindquist blocks. In order to extend through the horizon at $r_X$, we define 
\begin{align}
\kappa_X=\frac{1}{2\chi(r_X^2+a^2)}\vert\partial_r\Delta_r\vert_{r=r_X}\, .
\end{align}
We can then construct Kruskal-type coordinates. On $\rI$, they are given by
\begin{align}
U_+&=-e^{-\kappa_+u_+}\, , & V_+&= e^{\kappa_+v_+}\, , & U_c&= e^{\kappa_c u_c}\, , &  V_c&=-e^{-\kappa_cv_c}\, .
\end{align}
As a result, $\rI$ corresponds to $\{U_+<0,V_+>0\}\times \bS^2_{\theta_+,\varphi_+}$ (or $\{U_c>0,V_c<0\}\times \bS^2_{\theta_c,\varphi_c}$). Since the metric remains finite and non-degenerate as $U_X,V_X\to 0$, one can extend the spacetime to the Kruskal block $M_X=\bR_{U_X}\times\bR_{V_X}\times\bS^2_{\theta_X,\varphi_X}$ \cite{Borthwick:2018}. We then have $M\cap M_+=\{V_+>0\}$ and $M\cap M_c=\{U_c>0\}$. The submanifolds $\mH=\{V_+=0\}\subset M_+$ and $\mHc=\{U_c=0\}\subset M_c$ will be used later to construct the Unruh state. $\mH$ consists of the three pieces $\mH^L=\{V_+=0, U_+>0\}$, $\mH^-=\{V_+=0, U_+<0\}$ and the bifurcation sphere $\mathcal{B_+}=\{U_+=V_+=0\}$, while $\mHc$ consists of $\mHc^R=\{U_c=0,V_c>0\}$, $\mHc^-=\{U_c=0,V_c<0\}$ and $\mathcal{B}_c=\{U_c=V_c=0\}$. Both $\mH$ and $\mHc$ are part of the manifold $\tilde M=M_+\cup M_c$, where the blocks $\rI$ in $M_+$ and $M_c$ are identified with each other.  Correspondingly, $M$ can be embedded into $\tilde M$.

\begin{figure}
\centering
\includegraphics[scale=1]{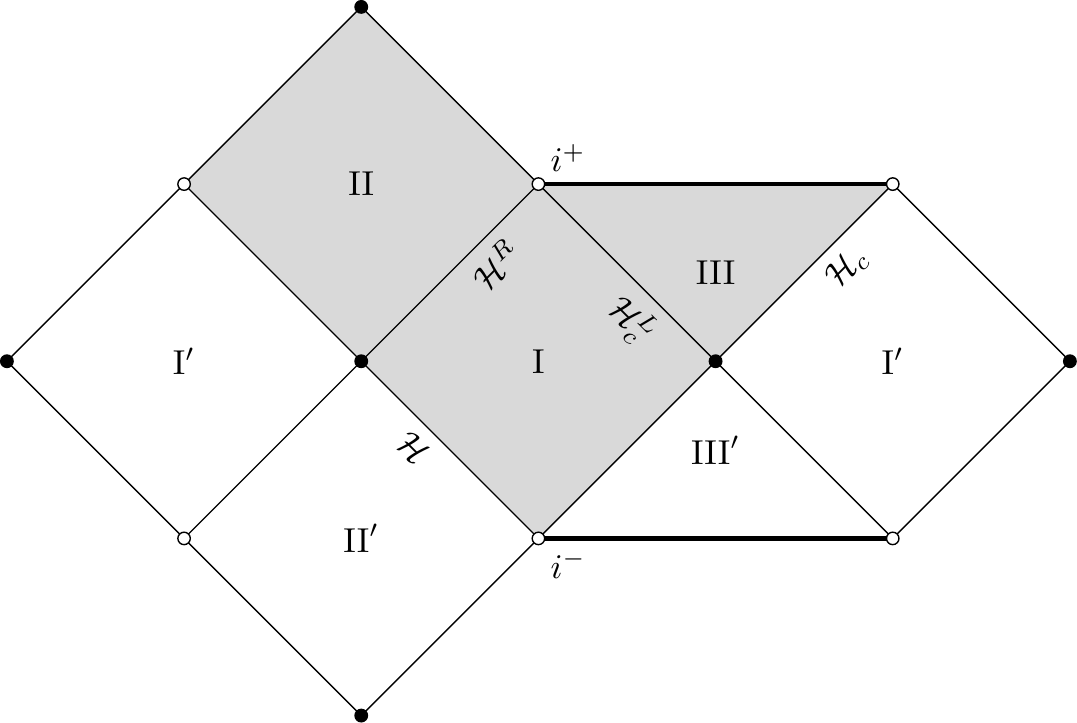}
\caption{Penrose diagram of the $(\theta, \varphi)=\text{const.}$-surface of the extended spacetime $\tilde M$. The gray area corresponds to $M$, the union of the blocks $\rI$, $\rII$ and $\rIII$. The prime indicates a reversal of the time orientation. The horizons $\mH^R$ and $\mHc^L$ are part of $M$, while the long horizons $\mH$ and $\mHc$ are the boundary of $M$ in $\tilde M$.}
\label{Fig:PD}
\end{figure}
The Penrose diagram for $M$ and $\tilde M$ is shown in Fig.~\ref{Fig:PD}.

Before moving on to the scalar field, let us show some results that will become important later on. The first one is a result on covectors on $\mH$, parametrized in Kruskal coordinates. Note that a covector $k\in T_x^* M$ will be called future directed, or future pointing, if $\langle k,v\rangle>0$ for any timelike vector $v$ in the future lightcone $V_x^+$.

\begin{lem}
\label{lem:Hgeo}
Denote by $\psi_+:M_+\to \bR^2\times \bS^2$ the coordinate map of the +-Kruskal coordinates.
If $(U_+,\theta,\varphi_+,\xi,\sigma_\theta,\sigma_\varphi)\in T^*(\bR\times\bS^2)$, then there is a unique $\eta(\xi,\sigma_\theta,\sigma_\varphi)\in\bR$ such that $\psi_+^*(U_+,0,\theta,\varphi_+,\xi,\eta,\sigma_\theta,\sigma_\varphi)$ is null and does not lie in the conormal space of $\mH$, $N^*(\mH)$, iff $\xi\neq 0$. In this case $\psi_+^*(U_+,0,\theta,\varphi_+,\xi,\eta,\sigma_\theta,\sigma_\varphi)$ is future pointing iff $\xi>0$. 
\end{lem}
\begin{proof}
On $\mH$, the metric takes the form (see e.g. \cite{Borthwick:2018})
\begin{align*}
g=g_{VV} \td V_+^2 +2g_{UV}\td U_+ \td V_+ +2 g_{V\varphi} \td \varphi_+ \td V_+ +g_{\theta\theta} \td \theta^2 + g_{\varphi\varphi}\td\varphi_+^2
\end{align*}
for some smooth functions $g_{\mu\nu}$, of which $g_{UV}<0$, $g_{\theta\theta}>0$ and $g_{\varphi\varphi}>0$\footnote{Except for on the axis where $\sin^2\theta=0$. However, the metric remains invertible there, as can be seen by going to appropriate coordinates, compare \cite{Borthwick:2018} and \cite[Rem. 3.3]{Hintz:2015}.}. Thus
\begin{align*}
&g^{-1}((\xi,\eta,\sigma_\theta,\sigma_\varphi),(\xi,\eta,\sigma_\theta,\sigma_\varphi))\\
&=\frac{1}{g_{UV}^2g_{\varphi\varphi}}\left(g_{V\varphi}^2-g_{VV}g_{\varphi\varphi}\right)\xi^2+\frac{2\xi\eta}{g_{UV}}-\frac{2g_{V\varphi}\xi\sigma_{\varphi}}{g_{UV}g_{\varphi\varphi}}+\frac{\sigma_\theta^2}{g_{\theta\theta}}+\frac{\sigma_{\varphi}^2}{g_{\varphi\varphi}}\, .
\end{align*}
If $\xi=0$, then this can only be zero if also $\sigma_\theta=0$ and $\sigma_{\varphi}=0$. But then $(\xi, \eta,\sigma_\theta,\sigma_\varphi)=(0,\eta,0)\in N^*(\mH)$. Hence we must have $\xi\neq 0$. And in turn, if $\xi\neq 0$, then $(\xi, \eta,\sigma_\theta,\sigma_\varphi)$ cannot be in $N^*(\mH)$. Moreover, the null condition can be solved for $\eta(\xi,\sigma_\theta,\sigma_\varphi)$, and since it is linear in $\eta$, there will be a unique solution.
The rest follows from the fact that $\partial_{U_+}$ is a future-pointing null vector on $\mH$, and since $\xi\neq 0$ we have $\langle(\xi,\eta,\sigma_\theta,\sigma_\varphi),\partial_{U_+}\rangle =\xi\neq 0$. By introducing normal coordinates one can then show that $(\xi,\eta,\sigma_\theta,\sigma_\varphi)$ is future pointing iff $\xi>0$.
\end{proof}

The same proof with $U\leftrightarrow V$ and $+\leftrightarrow c$ shows the corresponding statement for covectors on $\mHc$.

Next, we show two results based on the behaviour of null geodesics in $M$. 
There are three constants of motion: The energy $E=-g( \gamma^\prime,\partial_t)$, the angular momentum in the direction of the rotation axis $L=g( \gamma^\prime,\partial_\varphi)$, and the Carter constant $K$ \cite{Carter:1968}. Here, $\gamma^\prime$ is the tangent vector of the geodesic $\gamma$.

With the help of these constants, the geodesic equation can be separated and written as \cite{Hackmann:2010, Salazar:2017, Borthwick:2018}
\begin{subequations}
\label{eq:geod eqn}
\begin{align}
    \rho^4\left(\frac{\td r}{\td \tau}\right)^2&=\chi^2\left[(r^2+a^2)E-aL\right]^2-K\Delta_r\equiv R(r)\\
    \rho^4\left(\frac{\td\theta}{\td\tau}\right)^2&=K\Delta_\theta-\chi^2\left[\frac{L}{\sin\theta}-aE\sin\theta\right]^2\equiv \Theta(\theta)\\
    \rho^2\frac{\td t}{\td\tau}&=\frac{\chi^2(r^2+a^2)\left[(r^2+a^2)E-aL\right]}{\Delta_r}+\frac{\chi^2a(L-Ea\sin^2\theta)}{\Delta_\theta}\\
    \rho^2\frac{\td \varphi}{\td\tau}&=\frac{\chi^2a\left[(r^2+a^2)E-aL\right]}{\Delta_r}+\frac{\chi^2a\left(E-\frac{L}{a\sin^2\theta}\right)}{\Delta_\theta}
\end{align}
\end{subequations}
for light-like geodesics, which entails $K\geq 0$. One can convince oneself that 
\begin{align*}
    \frac{\td r}{\td\tau}=0&\Leftrightarrow R(r)=0\, ,\\
    \frac{\td r}{\td \tau}=0\text{ and }\frac{\td^2 r}{\td \tau^2}=0&\Leftrightarrow R(r)=0 \text{ and }\partial_rR(r)=0\, .
\end{align*}

With this, we can show the following lemma:

\begin{lem}
\label{lem:region} There exists a $\lambda_0>0$ and an $a_0>0$, such that for all $0<\lambda<\lambda_0$ and any $0<a<a_0$, any inextendible null geodesic on $M$ that does not approach $\mH$ or $\mHc$ in the past must intersect the region in which the vector fields
 \begin{align}
\partial_{t_X}=\partial_t+\frac{a}{r_X^2+a^2}\partial_\phi=\partial_{u_X}+\partial_{v_X} \, ,\quad X\in\{+,c\}\, ,
\end{align} 
 are both timelike.
\end{lem}

\begin{proof}
First, let us note that many of the results of \cite{ONeill:1995, Gerard:2020} on the null geodesics on Kerr can be extended to Kerr-de Sitter, see for example the results in \cite{Hackmann:2010, Salazar:2017, Borthwick:2018}. The resulting description of the null geodesics on Kerr-de Sitter is relayed to App.~\ref{sec:A2}. With these results, one finds that it is sufficient to consider null geodesics in region $\rI$. For such geodesics, there are two possibilities not to approach $\mH$ or $\mHc$ in the past: One is that $R(r)$ has two distinct zeros in $r_+<r<r_c$, in between which it is positive. In this case $r(\tau)$ will oscillate between the two zeros. This cannot happen due to the form of $R$. The other is that $R(r)$ has a double root $r_0$. In this case $r(\tau)=r_0$ for all $\tau$, or $r_0$ is approached asymptotically.

Hence we look for double roots of $R(r)$. Let us first assume $E=0$. In this case, 
\begin{align*}
\Theta(\theta)=\frac{K}{\sin^2\theta}\left(-a^2\lambda \cos^4\theta-(1-a^2\lambda)\cos^2\theta+1-\frac{L^2\chi^2}{K}\right)
\end{align*}
and hence there exists no solution for the geodesic if $\tfrac{L^2\chi^2}{K}>1$, since $a^2\lambda<1$ in the whole parameter range. The condition for the double root of $R(r)$ can be written as $\partial_r\Delta_r(r_0)=0$ and $\Delta_r(r_0)=a^2\frac{L^2\chi^2}{K}$. By the above, this needs to be smaller than or equal to $a^2$. By choosing $\lambda$ smaller than $\sim 0.0332$, one can ensure that this condition is not met and that there are no double roots of $R(r)$ with $E=0$.

Hence, we can restrict ourselves to the case $E\neq 0$. Introducing the rescaled $l=L/E$ and $k=K/\chi^2E^2$,  one can then write
\begin{subequations}
\begin{align}
R(r)=\chi^2E^2\left(\beta r^4+\gamma r^2+2kr-a^2q\right)\\
\Theta(\theta)=\frac{\chi^2E^2}{\sin^2\theta}\left(-a^2\beta \cos^4\theta+\gamma \cos^2\theta+q\right)
\end{align}
\end{subequations}
with
\begin{align*}
\beta&=1+\lambda k\, ,   & \gamma&=2a(a-l)-k(1-a^2\lambda)\, , & q&=k-(a-l)^2\, .
\end{align*}
It is then easy to see that if $q<0$, $\gamma$ needs to be positive, since otherwise $\Theta(\theta)$ is negative for any $\theta$ and no solution for the geodesic exists. But in this case all coefficients in the polynomial $R(r)$ are positive, so $R(r)>0$ for all $r>0$. Hence, there cannot be a double zero of $R(r)$ in $r>0$. This implies that $q$ needs to be non-negative.

Next, one can take the conditions for the double zero of $R(r)$ and solve them for $l$ and $k$. One finds
\begin{align*}
l&=\frac{\Delta_r^\prime (r^2+a^2)-4r\Delta_r}{a\Delta_r^\prime}\vert_{r=r_0}\, , &
k&=\frac{16r^2\Delta_r}{\Delta_r^{\prime 2}}\vert_{r=r_0}\, ,
\end{align*}
where a prime denotes a derivative with respect to $r$.
From this, one then finds
\begin{align}
q&=\frac{r^2}{a^2\Delta_r^{\prime 2}}\left(16 \Delta_r(a^2-\Delta_r)+r\Delta_r^\prime(8\Delta_r-r\Delta_r^\prime)\right)\vert_{r=r_0}\\\nonumber
&=\frac{4r^3}{a^2\Delta^{\prime 2}_r}\left(4a^2-r(r-3)^2-a^2\lambda r^2(2(r+3)+a^2\lambda r)\right)\vert_{r=r_0}
\end{align}
This form for $q$ is very similar to the one found in \cite{Gerard:2020} for the Kerr spacetime. 
Notice that the terms proportional to $\lambda$ enter with a minus sign. Hence they reduce the range of $r$ for which the expression in brackets is positive. One finds that the double roots of $R(r)$ must either lie in $r\in[0,r_1]$ for some $r_1<r_+$, which is not of interest to us, or in $r\in\left[3-\frac{2\sqrt{1-27\lambda}}{\sqrt{3}}a+\mathcal{O}(a^2),3+\frac{2\sqrt{1-27\lambda}}{\sqrt{3}}a+\mathcal{O}(a^2)\right]$, compare \cite[Lemma C.1]{Gerard:2020}.

The vector fields $\partial_{t_X}$ satisfy
\begin{align*}
g( \partial_{t_X},\partial_{t_X})=\frac{a^2\sin^2\theta\Delta_\theta (r_X^2-r^2)^2-\Delta_r \rho_X^4}{\chi^2\rho^2(r_X^2+a^2)^2}\, ,
\end{align*}
where $\rho_X=\rho(r_X,\theta)=r_X^2+a^2\cos^2\theta$. The numerator is monotonously decreasing in $\cos^2\theta$, and the denominator is always positive. Hence we can estimate
\begin{align*}
\chi^2 \rho^2 (r_X^2+a^2)^2 g( \partial_{t_X},\partial_{t_X})&\leq  a^2(r_X^2-r^2)^2 -\Delta_r r_X^4 \\
&\leq (1-27\lambda)r_X^4\vert_{a=0}\left[-3+\frac{8\sqrt{1-27\lambda}}{\sqrt{3}}a\right]+\mathcal{O}(a^2)\, ,
\end{align*}
where we took into account that $(r-3)$ is of order $a$ for any possible value of $r_0$.

Hence, for $\lambda\lesssim 0.332$, by a continuity argument as in \cite{Gerard:2020}, there must be some $a_0>0$ such that $\langle \partial_{t_X},\partial_{t_X}\rangle\vert_{r=r_0}<0$ for all possible values of $r_0$ as long as $0 \leq a < a_0$.

 We have also tested this numerically by checking that $\chi^2 \rho^2 (r_X^2+a^2)^2 g( \partial_{t_X},\partial_{t_X})<0$ for both $X$ for all allowed values of $r_0$ for fixed $\lambda$, varying $\lambda$ over its allowed range. We find that for all allowed values of $\lambda$, $a_0\sim 0.7$, with only a percent-level variation of that value. 
\end{proof}

\begin{figure}
\centering
\includegraphics[scale=0.6]{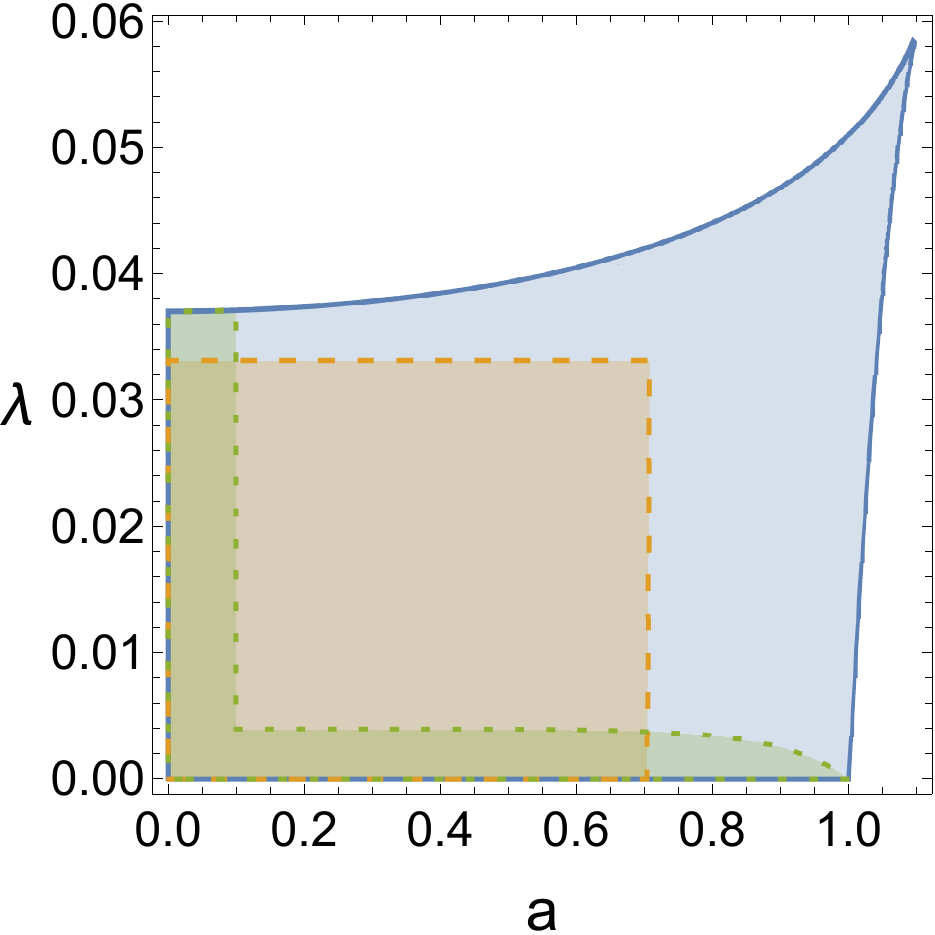}
\caption{The parameter region in the $(a,\lambda)$ plane. The region surrounded by the solid line is the subextremal range of the parameters. The region surrounded by the dashed line indicates approximately where Lemma \ref{lem:region} is valid. The region surrounded by the dotted line is an approximate indication for the parameter region in which mode stability for the scalar wave equation has been proven, see \cite{Dyatlov:2010, Hintz:2021}, and in particular Fig. 1.1 in \cite{Hintz:2021}.}
\label{fig:parRegs} 
\end{figure}

Figure \ref{fig:parRegs} depicts the parameter region allowed by the subextremality condition, as well as the approximate parameter regions in which the above Lemma and mode stability \cite{Dyatlov:2010, Hintz:2021} hold. The above Lemma is valid in a large portion of the parameter space. However, it cannot cover the case of rapidly rotating black holes with a small cosmological constant, which would be very interesting to study and for which mode stability results have been obtained recently \cite{Hintz:2021}. Hence, a different strategy would be necessary to prove the Hadamard property of the Unruh state in this regime.

In addition to the Lemma above, the analysis of the null geodesics on Kerr-de Sitter also allows us to show

\begin{prop}
$M$ and $\tilde M$ are globally hyperbolic.
\end{prop}
\begin{proof}
Thouroughly checking  the arguments made in \cite{ONeill:1995} for the case of a Kerr spacetime, we find that the results of \cite{ONeill:1995} and \cite[App. C]{Gerard:2020} on the behaviour of null geodesics in Kerr extend to Kerr- de Sitter with only minimal modifications; see also \cite{Borthwick:2018}. 
 
In addition, let us note that the function $x(r)=r_*(r)-r$ is strictly monotonic on $(r_+,r_c)$ and ranges from $-\infty$ at $r_+$ to $\infty$ at $r_c$. As a result, for any $T\gg 1$, there will be a unique solution $r_T$ of $x(r)=-T$ near $r_+$ and a unique solution $r_T^\prime$ of $x(r)=T$ near $r_c$. We may choose $T$ large enough such that $r_T<r_0<r_T^\prime$ for any double root $r_0$ of $R(r)$. We then set $u_T$ to be 
\begin{align*}
u_T=\begin{cases}u_c+r+T  & :\;r_T^\prime\leq r\\
t  & :\;r_T<r<r_T^\prime\\
v_++T-v(r)  & :\;r\leq r_T \end{cases}\, ,
\end{align*}
with 
\begin{align*}
v^\prime(r)&=1+\chi(r) \frac{1}{r-r_-}\,, & v(r_+)&=r_+\, ,
\end{align*}
and some $\chi\in C^\infty(\bR)$, $\chi=1$ on $(-\infty, r_-+\epsilon]$, and $\chi=0$ on $\left(1/2 (r_++r_-),\infty\right)$, see \cite[App. C.6.2]{Gerard:2020}. Then one can explicitly check that $\nabla u_T$ is timelike over the whole range of $r_-<r<\infty$:
on $\{r_T<r<r_T^\prime\}$, one finds 
\begin{align*}
g^{-1}(\td u_T, \td u_T)=g^{tt}\leq \frac{\chi^2}{\rho^2}\left(a^2-\frac{(r^2+a^2)^2}{\Delta_r}\right)=-\frac{\chi^2}{\rho^2}\frac{\chi r^4+\chi r^2a^2+2a^2 r}{\Delta_r}<0\, ,
\end{align*}
where the inequality follows from the fact that $\Delta_r>0$ in this region.

For $\{r_-\leq r\leq r_T\}$, we use the metric in $KdS*$-coordinates \cite{Borthwick:2018}, combined with the fact that $v^\prime(r)\geq 1$ and $v^\prime(r)=1$ when $\Delta_r \geq 0$. Similarly, on $\{r_T^\prime\leq r\}$, we combine the inverse metric in $*KdS$-coordinates \cite{Borthwick:2018} with the fact that $v(r)=r$. In both cases, we find
\begin{align*}
g^{-1}(\td u_T, \td u_T)\leq \frac{1}{\rho^2}\left(-\lambda r^4-(1+3\lambda a^2)r^2-2r+\lambda^2a^6\right)\, .
\end{align*}
The term in the brackets is a polynomial in $r$ with a single root $r_0$ in $r>0$. Moreover, the polynomial is negative for all $r>r_0$. One can use $\Delta_r(r_-)=0$ to reduce the terms in the bracket at $r=r_-$ to $\chi(\chi-2)a^2-2\chi r^2$, which is strictly negative in the whole range of spacetime parameters under consideration since $1<\chi<2$. Therefore, $\td u_T$ is time-like on $M$, compare also \cite[App. C.6.2]{Gerard:2020}. Moreover, for any inextendible future-directed null geodesic $\gamma$ one has $\sup_\gamma u_T=\infty$ and $\inf_\gamma u_T=-\infty$ by the extension of the results of \cite{ONeill:1995,Gerard:2020} to Kerr-de Sitter described in App.~\ref{sec:A2} and \cite{Borthwick:2018}.
Since $u_T$ thus satisfies the conditions of \cite[Cor. C.7]{Gerard:2020}, this shows that the spacetime $M$ is globally hyperbolic. We also notice that $\Sigma_{n,t_0}\equiv\{u_n=t_0\}$ is a family of Cauchy surfaces of $M$ converging to $\mH^L\cup \{t=t_0, r_+<r<r_c\}\cup \mHc^R$ as $n\to\infty$.

In addition, we may adapt \cite[Prop. C.12]{Gerard:2020}, by choosing 
\begin{align*}
\Sigma_M=\{U_+=-V_+\}\sqcup \{U_c=-V_c\}/\sim\, ,
\end{align*}
where $\sim$ is the identification of $\rI\subset M_+$ with $\rI\subset M_c$ in $\tilde M$. Then repeating the proof of \cite[Prop. C.12]{Gerard:2020} for this hypersurface, we see by direct inspection that it is achronal and, using the results collected in App.~\ref{sec:A2}, that any inextendible future-directed null geodesic must enter $I^+(\Sigma_M)$ and $I^-(\Sigma_M)$ \cite{ONeill:1995,Borthwick:2018,Gerard:2020}.
By \cite[Thm. C.6]{Gerard:2020}, $\tilde M$ is a globally hyperbolic manifold. 
\end{proof}

\section{The scalar field}
\label{sec:scalar field}
In this work, we consider the quantization of a real scalar field $\Phi$ satisfying the Klein-Gordon equation 
\begin{align}
\label{eq:KGE}
\mathcal{K}\Phi&=0\, , \quad \mathcal{K}=\nabla_a\nabla^a-m^2\, ,
\end{align}
where $m>0$ is a constant and $\nabla_\mu$ is the covariant derivative on $\tilde M$. Since $\tilde M$ is globally hyperbolic, there are unique retarded and advanced fundamental solutions $E^\pm:C_0^\infty(\tilde M)\to C^\infty(\tilde M)$ for the Klein-Gordon operator $\mathcal{K}$ on $\tilde M$. Here and in the following, $C^\infty_{(0)}(N)$ denotes the space of smooth, complex (and compactly supported) functions on $N$. The commutator function $E=E^+-E^-:C_0^\infty(\tilde M)\to S(\tilde M)$ maps compactly supported functions to the space of solution to the Klein-Gordon equation with compact support on spacelike hypersurfaces, which we denote $S(\tilde M)$. This space can be equipped with a symplectic form 
\begin{align}
\sigma(\phi,\psi)=\int\limits_\Sigma (\phi\nabla_a\psi-\psi\nabla_a\phi)n_\Sigma^a\td vol_\gamma\, ,
\end{align}
where $\Sigma$ is any piecewise smooth spacelike Cauchy surface, $n^a_\Sigma$ its future pointing normal vector and $\td vol_\gamma$ the volume element associated to the induced metric $\gamma$ on $\Sigma$. Note that $\sigma$ is independent of the choice of Cauchy surface by Gauß's law \cite{Dimock:1980}. It can be shown \cite{Dimock:1980} that 
\begin{align}
\label{eq:SymplMorph}
E(f,g)=\int\limits_{\tilde M} f(x) E(g)(x) \td vol_g(x)=\sigma(E(f),E(g))\, ,
\end{align}
and hence $E:C_0^\infty(\tilde M)/\mathcal{K}(C_0^\infty(\tilde M))\to S(\tilde M)$ is a symplectomorphism.
The same structure can be constructed for $M$ by restricting $E:C_0^\infty(M)\to S(M)\subset C^\infty(M)$.

We can then define the algebra of observables in the following way, see for example \cite{Hollands:2019,Fewster:2015}:

\begin{defi} The algebra of observables for the free scalar field, $\mathcal{A}$, is the free *-algebra generated by the unit element $\mathbf{1}$ and the elements $\Phi(f)$, $f\in C_0^\infty(M)$, subject to the relations
\begin{itemize}
\item {\bf Linearity} $\Phi(\alpha f+ g)=\alpha \Phi(f)+ \Phi(g)\quad \forall f,g\in C_0^\infty(M)$, $\alpha \in \bC$
\item {\bf Klein-Gordon equation} $\Phi(\mathcal{K}f)=0 \quad\forall f\in C_0^\infty(M)$
\item {\bf Hermiticiy} $(\Phi(f))^*=\Phi(\bar f)\quad \forall f\in C_0^\infty(M)$
\item {\bf Commutator property} $\left[\Phi(f),\Phi(g)\right]=iE(f,g)\mathbf{1} \quad\forall f,g\in C_0^\infty(M)$
\end{itemize}
\end{defi}

\begin{defi} A state on $\mathcal{A}$ is a linear map $\omega : \mathcal{A} \to \bC$, such that $\omega ( \mathbf{1})=1$ and\\ $\omega(A^*A)\geq 0$ $ \forall A\in \mathcal{A}$.
\end{defi}

Any state will be determined by its n-point functions
\begin{align*}
W_n^\omega(f_1,\dots,f_n)=\omega(\Phi(f_1)\dots\Phi(f_n))\, .
\end{align*}
A particular class of states are the so-called quasi-free or Gaussian states. They have the property that $W_n^\omega=0$ for $n$ odd and $W_n^\omega$ for $n$ even can be expressed in terms of the two-point function $W_2^\omega$ with the help of Wick's formula. Hence, a quasi-free state is completely determined by its two-point function. Turning the argument around, for a bi-distribution $w\in \mathcal{D}^\prime(M\times M)$ to be the two-point function of a quasi-free state on the algebra $\mathcal{A}$, it must satisfy
\begin{itemize}
\item {\bf Weak bi-solution} $w( \mathcal{K} (f) \otimes  g) = w(f \otimes \mathcal{K} (g) )=0 \quad \forall  f,g\in C_0^\infty(M)$
\item {\bf Positivity} $w(\bar f \otimes f)\geq 0\quad \forall f \in C_0^\infty(M)$
\item {\bf Commutator property} $w(f \otimes g)-w(g \otimes f)=i E(f,g) \quad \forall f,g\in C_0^\infty(M)$.
\end{itemize}

\subsection{The wavefront set and the Hadamard property}
For a state to be considered physically reasonable, one usually also demands it to be of Hadamard type. Radzikowski \cite{Radzikowski:1996} showed that in the case of a quasi-free state, the original formulation \cite{Kay:1988} of this condition is equivalent to a condition on the wavefront set of the two-point function, $W_2^\omega(x,y)$. So let us first introduce the wavefront set.

We will denote by
\begin{align}
\widehat{f}(k)=(2\pi)^{-n/2}\int\limits_{\bR^n}e^{ik\cdot x}f(x)\td^nx
\end{align}
the Fourier-Plancherel transform of $f\in\mathcal{E}^\prime(\bR^n)$.

\begin{defi}
\label{def.:WFS}
 Let $u\in \mD^\prime(\bR^n)$ a distribution, i.e. $u:C_0^\infty(\bR^n)\to \bC$ is a linear map which is continuous in the inductive limit topology on the test functions $C_0^\infty(\bR^n)$. Let $(x,k)\in \bR^n\times(\bR^n \backslash \{ 0 \} )$. Then $(x,k)$ is a direction of rapid decrease for $u$ if there exists a function $\chi\in C_0^\infty(\bR^n)$, $\chi(x)\neq 0$ and an open conic neighbourhood of $k$, $V_k\subset \bR^n \backslash \{ 0 \}$, i.e. if $k^\prime\in V_k$, then $\lambda k^\prime\in V_k$ for all $\lambda>0$, so that for any $N\in \bN$ there is a $C_N>0$ with \cite[Sec. 8.1]{Hoermander}
\begin{align}
\vert\widehat{\chi u}\vert(\xi)\leq C_N (1+\vert\xi\vert)^{-N} \quad \forall \xi\in V_k\, ,
\end{align}
i.e. the function $\widehat{\chi u}$ is rapidly decreasing in $\xi\in V_k$.
The wavefront set of $u$ is the set of all $(x,k)\in \bR^n\times(\bR^n \backslash \{ 0 \} )$ which are not of rapid decrease for $u$.
\end{defi}

A different characterization of the wavefront set due to \cite[Prop.2.1]{Verch:1998}, which we will use later, is
\begin{prop}[\cite{Verch:1998}]
\label{prop:V,Prop2.1}
Let $(x,k)\in \bR^n\times(\bR^n\backslash\{0\})$, $u\in \mD^\prime(\bR^n)$. Then $(x,k)\notin \WF(u)$ iff there exist an open neighbourhood $V\subset(\bR^n\backslash\{0\})$ of $k$,  some $h\in C_0^\infty(\bR^n)$ with $h(0)=1$ and some $g\in C_0^\infty(\bR^n)$: $\hat g(0)=1$ such that $\forall p\geq 1$, $\forall N\in \bN$, $\exists C_N>0$, $\lambda_N>0$ such that
\begin{align}
\sup\limits_{k^\prime\in V}\left\vert \int e^{i\lambda^{-1}k^\prime\cdot y}h(y)u\left(g(\lambda^{-p}(\cdot-x-y))\right)\td^ny\right\vert<C_N\lambda^N\quad \forall 0<\lambda<\lambda_N\, .
\end{align}
\end{prop}

If $u\in \mD^\prime(N)$, where $N$ is an arbitrary smooth manifold, we can define its wavefront set $\WF(u)\subset T^*N\backslash o$, where $o$ is the zero section, such that its restriction (in the base variable) to a coordinate patch $N_\psi\subset N$ with the coordinate map $\psi: N_\psi\to \mathcal{U}_\psi\subset\bR^n$ is \cite[Thm. 8.2.4]{Hoermander}
\begin{align}
\WF(u)\vert_{N_\psi}=\psi^*\WF(u\circ \psi^{-1})=\{(x,{}^t\td \psi(x)k):(\psi(x),k)\in\WF(u\circ \psi^{-1})\}\, .
\end{align}
For a distribution $u\in \mD^\prime(N\times N)$, we will also define the primed wavefront set
\begin{align}
\WF^\prime(u)=\{(x_1,k_1;x_2,k_2)\in T^*(N\times N)\backslash o: (x_1,k_1;x_2,-k_2)\in \WF(u)\}\, .
\end{align}

Let us now come back to the Hadamard property.
\begin{defi}
A quasi-free state $\omega$ on $\mathcal{A}$ has the Hadamard property if it satisfies the microlocal spectrum condition \cite{Radzikowski:1996}
\begin{subequations}
\begin{align}
&\WF^\prime(W^\omega_2)=\mC^+\\
&\mC^{\pm}=\left\{(x_1,k_1;x_2,k_2)\in T^*(M\times M)\setminus o:(x_1,k_1)\sim (x_2,k_2),\pm k_1 \text{ f.-dir.}\right\}\,. 
\end{align}
\end{subequations}
Here, $(x_1,k_1)\sim (x_2,k_2)$ means that $x_1$ and $x_2$ can be connected by a null geodesic, to which $k_1$ is cotangent at $x_1$ and $k_2$ is the same as $k_1$ parallel transported to $x_2$ along the geodesic. Recall that a covector $\xi\in T^*_xM$ is future-directed, if $\langle \xi,v\rangle>0$ for all timelike  $v\in V^+_x$.
\end{defi}

\section{The Unruh state on slowly-rotating Kerr-de Sitter}
\label{sec:Unruh}
In this section, we will specify the two-point function of the Unruh state on the Kerr-de Sitter spacetime $M$ and show that it indeed satisfies the conditions for being the two-point function of a state on $\mathcal{A}$. The two-point function of the state will be a combination of the Kay-Wald two-point function \cite{Kay:1988} on the past event horizon $\mH$ and the past cosmological horizon $\mHc$.

We will use the notation $\mH_+=\mH$, $L_+=U_+$, $l_+=u_+$, $L_c=V_c$, $l_c=v_c$ and $\Omega_X=(\theta,\varphi_X)$. We will denote by $\td^2 \Omega_X$ the volume element of $\bS^2_{\theta,\varphi_X}$, and we will identify $\mH_X=\bR_{L_X}\times \bS^2_{\theta,\varphi_X}$ and $\mH_X^-=\bR_{l_X}\times\bS^2_{\theta,\varphi_X}$ unless specified otherwise.

\begin{defi}
\label{def:2ptfct}
For $\phi,\psi\in C^\infty_0(\mH_X)$, we define 
\begin{align}
A_X(\phi,\psi)=-\lim\limits_{\epsilon\to 0}\frac{r_X^2+a^2}{\chi \pi}\int \frac{\phi(L_X,\Omega_X)\psi(L_X^\prime,\Omega_X)}{(L_X-L_X^ \prime-i\epsilon)^2}\td L_X\td L_X^\prime\td^2\Omega_X\, ,
\end{align}
with $\chi$ as in \eqref{eq:D_t,rho,chi}.
The two-point function of the Unruh state is then defined as
\begin{align}
\label{eq:2ptfct}
    w(f,h)=& w_+(f,h)+w_c(f,h)\\\nonumber
     =& A_+(E(f)\vert_{\mH},E(h)\vert_{\mH})+A_c(E(f)\vert_{\mHc},E(h)\vert_{\mHc})\\\nonumber
   =& -\lim\limits_{\epsilon\to 0^+}\frac{r_+^2+a^2}{\chi\pi}\int\frac{E(f)\vert_{\mH}(U_{+},\Omega_+)E(h)\vert_{\mH}(U_+^\prime,\Omega_+)}{(U_+-U_+^\prime-i\epsilon)^2}\td U_+\td U_+^\prime \td^2 \Omega_+\\\nonumber
    &-\lim\limits_{\epsilon\to 0^+}\frac{r_c^2+a^2}{\chi\pi}\int\frac{E(f)\vert_{\mHc}(V_c,\Omega_c)E(h)\vert_{\mHc}(V_c^\prime,\Omega_c)}{(V_c-V_c^\prime-i\epsilon)^2}\td V_c\td V_c^\prime \td^2 \Omega_c\, 
\end{align} 
for any two test functions $f,h\in C_0^\infty(M)$.
\end{defi}

\subsection{Well-definedness of the Unruh two-point function}
While $E(f)$ is compactly supported when restricted to any spacelike Cauchy surface of $M$ for any $f\in C_0^\infty(M)$, it is not compactly supported on the light-like hypersurfaces $\mH$ and $\mHc$. Hence the convergence of the integrals in \eqref{eq:2ptfct} is not automatic.
Thus, before we can show that \eqref{eq:2ptfct} is the two-point function of a state on $\mathcal{A}$, we need to demonstrate that it is indeed well-defined in the sense that the integrals converge. 

For the proof we will make use of the estimates in \cite{Hintz:2015}. However, their results only hold for $\vert a\vert\ll 1$ or $\lambda\ll 1/27$ and $\vert a\vert<1$, so that from now on we restrict ourselves to this parameter region \footnote{The reason is that the necessary mode stability results, in particular the presence of a spectral gap $\alpha>0$ for quasi-normal mode solutions of the massive wave equation, have only been proven by perturbation of the results on Schwarzschild-de Sitter ($a=0$) \cite{Dyatlov:2010} or Kerr ($\lambda=0$)\cite{Hintz:2021}. One would expect that mode stability holds in the whole subextremal regime, but this remains to be shown, see also \cite[Rem.3.6]{Hintz:2015}.}.

\begin{prop}
\label{prop:wellDef}
If $0<a\ll 1$ or $\lambda\ll 1/27$ and $0<a<1$, then $w(f,h)$ as defined in \eqref{eq:2ptfct} is a well-defined bi-distribution $w\in \mD^\prime(M\times M)$.
\end{prop}

\begin{proof}
 In \cite{Hintz:2015}, as also analysed in \cite[Thm. 4.4]{Hollands:2019}, the authors prove, after an application of the $t\to-t$, $\varphi\to-\varphi$ symmetry and Sobolev embedding, the estimate
\begin{align}
\label{eq:estGEn}
 \vert\partial^N E^-(f)\vert(t_*,r,\theta,\varphi_*)\leq C e^{\alpha t_*}\, , \quad \partial\in\{\partial_{t_*},\partial_r,\partial_\theta,\partial_{\varphi_*}\}
\end{align}  
for points sufficiently close to $i^-$. $\varphi_*$ corresponds to the $\varphi$-coordinate in the $KdS*$- ($*KdS$-) coordinates near $r_+$ $(r_c)$ \cite{Hintz:2015}. The coordinate $t_*$ corresponds to $t$ on  $(r_++\delta,r_c-\delta)$ for some small $\delta>0$ and approaches $u_+$ near $\mH^-$ and $v_c$ near $\mHc^-$ up to finite terms. This allows the estimates
\begin{align}
e^{\alpha t_*}\leq\begin{cases}
    \tilde{C}(\delta,\delta^\prime)e^{\alpha t} & r\in (r_++\delta^\prime,r_c-\delta^\prime)\\
    \tilde{C}(\delta,\delta^\prime)e^{\alpha u_+} & r\in [r_+,r_++\delta^\prime]\\
    \tilde{C}(\delta, \delta^\prime)e^{\alpha v_-} & r\in [r_c-\delta^\prime,r_c]
    \end{cases}\, ,
\end{align}
for points sufficiently close to $i^-$ for some $0<\delta^\prime<\delta$. The constants depend on the concrete implementation of $t_*$. $\alpha$ is the spectral gap of the Klein-Gordon operator on this spacetime.

As described in \cite{Hollands:2019}, the constants $C$ can be estimated by $C^\prime \norm{f}_{C^{m(N)}}$ using the Fredholm property of the Klein-Gordon operator derived in \cite{Hintz:2015}. Assuming that $\supp(f)\subset K$ for some compact region $K\subset M$, and that $V_i$, $i=1,\dots 4$ are linearly independent smooth vector fields on $K$,
\begin{align}
\norm{f}_{C^m}=\max\limits_{\vert\beta\vert\leq m}\sup\limits_{x\in K} \vert V^\beta f(x)\vert\, ,
\end{align}
where $\beta\in \bN^4$, $\vert\beta\vert=\beta_1+\beta_2+\beta_3+\beta_4$, and $V^\beta=\prod_iV_i^{\beta_i}$. The constant $C^\prime$ will depend on $K$.

Noting that $\partial_{t_*}\to \partial_{u_+}$ for $r\to r_+$ and $\partial_{t_*}\to\partial_{v_c}$ for $r\to r_c$ together with the relation between $L_X$ and $l_X$ then yield for $n=0,1$ and $N\in \bN$
\begin{subequations}
\begin{align}
\label{eq:estt}
   & \vert\partial^N E(f)\vert \leq C^\prime \norm{f}_{C^m} e^{\alpha t} \text{ on }  \{r_++\delta^\prime< r<    r_c-\delta^\prime\} \text{ with } \partial\in\{\partial_t, \partial_r,\partial_\theta,\partial_{\varphi_*}\}\\
\label{eq:estU}
    &\vert \partial_{U_+}^nE(f)\vert \leq C^\prime \norm{f}_{C^m} \vert U_+\vert^{-(n+\alpha/\kappa_+)} \text{ on } \{r_+\leq r\leq r_++\delta^\prime\}\\
\label{eq:estV}
   & \vert\partial_{V_c}^nE(f)\vert\leq C^\prime \norm{f}_{C^m} \vert V_c\vert^{-(n+\alpha/\kappa_c)} \text{ on }  \{r_c-\delta^\prime \leq r\leq r_c\}
\end{align}
\end{subequations}
for any $f\in C_0^\infty(M)$, sufficiently close to $i^-$, where $m\in \bN$ depends on $N$ or $n$ respectively.

In addition, by the support properties of $E$, there are constants $U_f$ and $V_f$ such that $\supp(E(f)\vert_{\mH})\subset\{U_+\leq U_f\}$ and  $\supp(E(f)\vert_{\mHc})\subset\{V_c\leq V_f\}$, and $U_f$, $V_f$ only depend on the support of $f$.

Now, let us consider the first part of \eqref{eq:2ptfct}, 
$w_+(f,h)$.
Utilizing the estimate \eqref{eq:estU}, we can integrate by parts twice to get
\begin{align*}
\vert w_{+}(f,h)\vert=& \lim\limits_{\epsilon\to 0}\left\vert \frac{r_+^2+a^2}{\chi} \int\limits_{\mathclap{\bR\times\bR\times \bS^2}} \partial_{U_+}E(f)\vert_{\mH}(U_+,\Omega_+) \partial_{U_+^\prime} E(h)\vert_{\mH}(U_+^\prime,\Omega_+) \right. \\
&\left.\times\vphantom{\frac{r_+^2+a^2}{\chi} \int\limits_{\bR\times\bR\times \bS^2}}\log(U_+-U_+^\prime-i\epsilon)\td U_+\td U_+^\prime\td^2\Omega_+\right\vert\, .
\end{align*}

Let us keep $\epsilon$ fixed for the moment, and let $U_0>0$ be a constant such that the estimate \eqref{eq:estU} holds for $U_+\leq -U_0$ for both $f$ and $h$. We define $I=[-U_0,\infty)$ and $I^\mathrm{c}=\bR\backslash I$ and split the integral into integrals  $A_j$, $j\in\{1,2,3,4\}$, over the regions $D_j$\footnote{If $U_f\leq-U_0$ or $U_h\leq-U_0$, the the corresponding parts of the integral just drop out.},
\begin{align*}
D_1&=I \times I\times \bS^2\, , & D_2&=I \times I^\mathrm{c}\times \bS^2\, , & D_3&=I \times I^\mathrm{c}\times \bS^2\, , & D_4&=I^\mathrm{c} \times I^\mathrm{c} \times \bS^2\, .
\end{align*}

\begin{figure}
\centering
\includegraphics[scale=1]{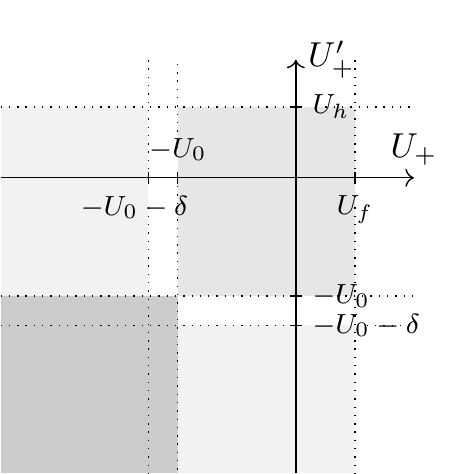}
\caption{The integration regions in the $U_+-U_+^\prime$-plane. The upper-right corner shows the support of the integrand in $D_1$. The lower-left corner indicates $D_4$. The light-gray region and the white stripe above and to the left of $D_4$ are $D_2$ and $D_3$.}
\label{fig:DOI} 
\end{figure}

The integration regions are indicated in figure~\ref{fig:DOI}.

On $D_1$, the integrand is supported on the compact subset $[-U_0,U_f]\times [-U_0,U_h]\times\bS^2$. We thus find
\begin{align*}
\vert A_1\vert\leq& C_1 \sup\limits_{I\times \bS^2} \vert\partial_{U_+} E(f)\vert_{\mH}\vert\sup\limits_{I\times \bS^2}\vert\partial_{U_+}E(h)\vert_{\mH}\vert\, \vert[-2U_0,U_f+U_h]\vert\\
&\times\norm{\log (y-i\epsilon)}_{L^1([-U_0-U_h,U_0+U_f])}
\end{align*} 
for some $C_1>0$. Note that $\log(\cdot-i\epsilon)\in L^1_{loc}(\bR)$, and that it converges for $\epsilon\to 0$ in $L^1_{loc}(\bR)$ to some $l\in L^1_{loc}(\bR)$. In addition, the suprema can be estimated by some $C^k$-norm of $f$ and $h$ due to the continuity of the causal propagator.

To estimate $\vert A_2\vert$, we further split $D_2$ into $D_2^a=I\times [-U_0-\delta, -U_0)\times \bS^2$ and\\ $D_2^b=I\times (-\infty, -U_0-\delta)\times \bS^2$, where $\delta>0$ is some constant. Then the term $A_2^a$ can be estimated similar to $A_1$ by
\begin{align*}
\vert A_2^a\vert\leq& C_2^a \sup\limits_{I\times \bS^2} \vert\partial_{U_+} E(f)\vert_{\mH}\vert\sup\limits_{[-U_0-\delta, -U_0)\times \bS^2}\vert\partial_{U_+}E(h)\vert_{\mH}\vert\,\vert[-2U_0-\delta,U_f-U_0]\vert\\
&\times\norm{\log (y-i\epsilon)}_{L^1([0, U_f+U_0+\delta])}\, .
\end{align*}
For $A_2^b$, we utilize that for any $c>0$, $\beta>0$, there is a constant $C_{c,\beta}>0$ such that\\ $\vert\log (y-i\epsilon)\vert\leq C_{c,\beta}\vert y\vert^\beta$ for all $\vert y\vert>c$. Together with the estimate \eqref{eq:estU} and the coordinate change $U_+^\prime\to -U_+^\prime$, we find
\begin{align*}
\vert A_2^b\vert \leq \tilde{C_2^b} \norm{h}_{C^{m(1)}} \sup\limits_{I\times \bS^2} \vert\partial_{U_+} E(f)\vert_{\mH}\vert \int\limits_{\mathclap{[-U_0,U_f]\times (U_0+\delta,\infty)}}\vert U_+^\prime\vert^{-1-\alpha/\kappa_+}\vert U_++U_+^\prime\vert^\beta\td U_+\td U_+^\prime\, .
\end{align*}
We can now choose $\beta=\alpha/2\kappa_+$ and estimate $\vert U_++U_+^\prime\vert\leq \vert U_+^\prime\vert\left(1+\tfrac{\vert U_f\vert}{U_0}\right)$ to get 
\begin{align*}
\vert A_2^b\vert &\leq \tilde{ \tilde{C_2}}^b \norm{h}_{C^{m(1)}} \sup\limits_{I\times \bS^2} \vert\partial_{U_+} E(f)\vert_{\mH}\vert\, \vert[-U_0,U_f]\vert \left(1+\frac{\vert U_f\vert}{U_0}\right)^{\tfrac{\alpha}{2\kappa_+}} \int\limits_{\mathclap{U_0+\delta}}^\infty \vert U_+^\prime\vert^{-1-\tfrac{\alpha}{2\kappa_+}} \td U^\prime_+\\
 &\leq C_2^b \norm{h}_{C^{m(1)}} \sup\limits_{I\times \bS^2} \vert\partial_{U_+} E(f)\vert_{\mH}\vert\, \vert[-U_0,U_f]\vert\left(1+\frac{\vert U_f\vert}{U_0}\right)^{\tfrac{\alpha}{2\kappa_+}}
\end{align*}

The term $A_3$ can be estimated in the same way as $A_2$.

Finally, we have for $A_4$ by a sign flip in both variables
\begin{align*}
\vert A_4\vert \leq \tilde C_4 \norm{f}_{C^{m(1)}} \norm{h}_{C^{m(1)}} \int\limits_{\mathclap{(U_0,\infty)\times(U_0,\infty)}} (U_+U_+^\prime)^{-1-\alpha/\kappa_+}\vert\log(U_+^\prime-U_+-i\epsilon)\vert\td U_+\td U_+^\prime\, .
\end{align*}
By \cite[Lemma 6.3]{Hollands:2000}, this integral is finite and converges for $\epsilon\to 0$ to some finite number. 

As a result, we find in the limit $\epsilon\to 0$
\begin{align}
\vert w_+(f,h)\vert\leq C(K)\norm{f}_{C^m}\norm{h}_{C^m}\, ,
\end{align}
where $K\subset M$ is a compact subset such that $\supp(f)\subset K$ and $\supp(h)\subset K$ and $m$ is chosen as the maximum of the different values for $m$ appearing in the estimates above.

By interchanging $U\leftrightarrow V$ and $+\leftrightarrow c$, the same estimates can be obtained for $w_c(f,h)$.
Hence for any $f,h\in C_0^\infty(K)$, where $K\subset M$ is some compact set, there is a $m\in \bN$, such that
\begin{align}
\label{eq:estiW}
\vert w(f,h)\vert\leq C(K) \norm{f}_{C^m}\norm{h}_{C^m}\, .
\end{align}
Thus $w(f,h)$ is a well-defined bi-distribution and by the Schwartz kernel theorem, its kernel $w(x,y)$ is in  $\mD^\prime(M\times M)$.
\end{proof}

As a result of the estimates \eqref{eq:estU} and \eqref{eq:estV} and their coordinate transform, we get that for any $f\in C_0^\infty(M)$, 
\begin{align}
&E(f)\vert_{\mH_X}\in S(\mH_X)\equiv \left\{\vphantom{ C_{\phi,N}(1+\vert L_X\vert)^{-(\alpha/\kappa_X)-N}}\phi \in C^\infty(\mH_X) : \exists L_\phi, C_{\phi,N}, N=0,1:\right.\\\nonumber
    & \phi(L_X,\Omega_X)=0\,\forall L_X\geq L_\phi\left.\text{ and } \vert\partial_{L_X}^N\phi(L_X,\Omega_X)\vert\leq C_{\phi,N}(1+\vert L_X\vert)^{-\frac{\alpha}{\kappa_X}-N}\right\}
\end{align}
and for any $f\in C_0^\infty(\rI)$
\begin{align}
\label{eq:S(H-)}
&E(f)\vert_{\mH_X}\in S(\mH_X^-)\equiv \left\{\vphantom{ C_{\phi,N}e^{-\alpha \vert l_X\vert}}\phi \in C^\infty(\mH_X^-): \exists l_{\phi}, C_{\phi,N}, N= 0,1:\right.\\\nonumber
    &\left. \phi(l_X,\Omega_X)=0\, \forall l_X\geq l_\phi \text{ and }\vert\partial^N_{l_X}\phi(l_X,\Omega_X)\vert\leq C_{\phi,N}e^{-\alpha \vert l_X\vert}\right\}\, .
\end{align}

Next, we can prove the necessary properties for $w(f,h)$ to define a two-point function of a state on the algebra $\mathcal{A}$. First of all, we notice that $E(\mathcal{K}(f))=0$, hence $w(f,h)$ is a weak bi-solution to the Klein-Gordon equation \eqref{eq:KGE}. 

To prove positivity, let us note that the results of \cite[sec. 3]{Dappiaggi:2009} can be translated to the present case by a careful adaptation of the appearing constants, see also \cite{Kay:1988}. In particular, the Hilbert space isomorphisms provided by \cite[Prop. 3.2 a)]{Dappiaggi:2009} and \cite[Prop. 3.3 a)]{Dappiaggi:2009} still hold if the constant $r_S^2$ in front of the integrals in $A_X$ (or $\lambda_{KW}$ in the notation of \cite{Dappiaggi:2009}) is replaced by $(r_X^2+a^2)/\chi$, and if, in \cite[Prop. 3.3 a)]{Dappiaggi:2009}, $(2r_S)^{-1}$ is replaced by $\kappa_X$ to accommodate for the different connection between $l_X$ and $L_X$:

\begin{prop}
\label{prop:FT version}
\begin{enumerate}
\item Equipping $C_0^\infty(\mH_X)$ with the hermitian sesquilinear form $A_X(\bar{\cdot},\cdot)$, the map
\begin{subequations}
\begin{align} 
F:C_0^\infty(\mH_X)&\to L^2(\bR_+\times\bS^2; \nu_X(\eta) \td\eta\td^2\Omega_X)\\
\phi&\mapsto F(\phi)=\left.(2\pi)^{-\tfrac{1}{2}}\int e^{i\eta L_X}\phi(L_X,\theta,\varphi_X)\td L_X\right\vert_{\{\eta\geq 0\}} \, ,
\end{align}
\end{subequations}
with $\nu_X(\eta)=2\eta (r_X^2+a^2)\chi^{-1}$, is an isometry and by continuity and linearity extends to a Hilbert space isomorphism mapping $\overline{(C_0^\infty(\mH_X),A_X(\bar{\cdot},\cdot))}$, the Hilbert completion of $(C_0^\infty(\mH_X),A_X(\bar{\cdot},\cdot))$, onto $L^2(\bR_+\times\bS^2; \nu_X(\eta)\td\eta\td^2\Omega_X)$ \cite[Prop. 3.2 a)]{Dappiaggi:2009}.
\item The map
\begin{subequations}
\begin{align}
\tilde{F}:C_0^\infty(\mH_X^-)&\to L^2(\bR\times\bS^2; \mu_X(\omega)\td\omega \td^2\Omega_X)\\
\phi&\mapsto \tilde F(\phi)=(2\pi)^{-\tfrac{1}{2}}\int e^{i\omega l_X}\phi(l_X,\theta,\varphi_X)\td l_X \, ,\\
\mu_X(\omega)&=\frac{r_X^2+a^2}{\chi}\frac{\omega e^{\pi\omega/\kappa_X}}{\sinh\left(\pi\omega/\kappa_X\right)}\, ,
\end{align}
\end{subequations}
is an isometry when $C_0^\infty(\mH_X^-)$ is equipped with the hermitian sesquilinear form $A_X(\bar{\cdot},\cdot)$. $\tilde F$ uniquely extends to a Hilbert space isomorphism from  $\overline{C_0^\infty(\mH_X^-)}$, as a Hilbert subspace of $\overline{(C_0^\infty(\mH_X), A_X(\bar{\cdot},\cdot))}$,  to\\ $ L^2(\bR\times\bS^2; \mu_X(\omega)\td\omega \td^2\Omega_X)$  \cite[Prop. 3.3 a)]{Dappiaggi:2009}.
\item Any $\phi\in S(\mH_X^-)$ can be identified with an element in $\overline{(C_0^\infty(\mH_X),A_X(\bar{\cdot},\cdot))}$ as described in \cite[Prop. 3.3 b)]{Dappiaggi:2009}.  Moreover, this identification is such that $\tilde F$ agrees with the Fourier-Plancherel transform in $l_X$.
\end{enumerate}
\end{prop}

Let us provide a brief sketch for the proof of the third point in Proposition~\ref{prop:FT version} as given in \cite[App. C]{Dappiaggi:2009}. The starting point for the proof is that for all $\phi\in S(\mH_X^-)$, $\phi$ and $\partial_{l_X}\phi$ lie in $L^2(\bR\times\bS^2, \td l_X\td \Omega_X)$. Therefore, $\phi$ lies in the Sobolev space $H^1(\bR\times\bS^2)_{l_X}$ of functions which are square integrable and have a square integrable $l_X$-derivative. It remains to show that iff $(\phi_n)_{n\in\bN},\, (\phi^\prime_n)_{n\in\bN}\subset C_0^\infty(\mH_X^-)$ are two sequences converging to $\phi\in S(\mH^-_X)$ in $H^1(\mH_X^-)$, then in $\overline{(C_0^\infty(\mH_X), A_X(\bar{\cdot},\cdot))}$ they are of Cauchy type and their difference converges to zero. The claim follows from the density of $C_0^\infty(\bR\times\bS^2)$ in $H^1(\bR\times\bS^2)_{l_X}$, an application of the Fourier-Plancherel transform and the isometry property of $\tilde F$.

We may now define the maps
\begin{subequations}
\begin{align}
&K_X: C_0^\infty(M)\to L^2(\bR_+\times\bS^2;\nu_X(\eta)\td\eta\td^2\Omega_X), \\\nonumber
& K_X(f) = F(\xi E(f)\vert_{\mH_X})+F((1-\xi)E(f)\vert_{\mH_X})\, ;\\
&K_X^\rI: C_0^\infty(\rI)\to L^2(\bR\times\bS^2;\mu_X(\omega)\td\omega\td^2\Omega_X)\, ,\\\nonumber &K_X^\rI(f)=\tilde{F}(E(f)\vert_{\mH_X})\, ,
\end{align}
\end{subequations}
where $\xi\in C^\infty(\bR_{L_X})$ is a real cutoff function such that $\xi(x)=1$ for $x>x_0$ and $\xi(x)=0$ for $x<x_1$ for some $x_1<x_0<0$. This is well-defined since $\xi E(f)\vert_{\mH_X}$ is compactly supported for any such $\xi$ and any $f\in C_0^\infty(M)$, while the second term can be understood by using the third part of the above proposition, see \cite{Dappiaggi:2009}. They satisfy

\begin{prop}
\label{prop:K_X}
The maps $K_X$ are independent of $\xi$, linear, and we can write
\begin{align}
w(f,h)=&\left\langle K_+(\overline{f}),K_+(h)\right\rangle_{L^2(\bR_+\times\bS^2;\nu_+(\eta)\td\eta\td^2\Omega_+)}\\\nonumber
&+\left\langle K_c(\overline{f}),K_c(h)\right\rangle_{L^2(\bR_+\times\bS^2;\nu_c(\eta)\td\eta\td^2\Omega_c)}
\end{align}
\end{prop}

\begin{proof}
Let  $\xi$ and $\xi^\prime$ be two functions satisfying the above conditions, and $f\in C_0^\infty(M)$. Then
\begin{align*}
 &F(\xi E(f)\vert_{\mH_X})+F((1-\xi)E(f)\vert_{\mH_X})- F(\xi^\prime E(f)\vert_{\mH_X})-F((1-\xi^\prime)E(f)\vert_{\mH_X})\\
& = F((\xi-\xi^\prime) E(f)\vert_{\mH_X})-F((\xi-\xi^\prime)E(f)\vert_{\mH_X})=0\, .
\end{align*}
Hence, the map is independent of the choice of $\xi$. The linearity follows from the fact that $E$, $F$, and multiplication by a bounded smooth functions are all linear maps.
Moreover, we have by the isometry property of $F$, for $f,h\in C_0^\infty(M)$,
\begin{align*}
w_X(f,h) =& A_X(E(f)\vert_{\mH_X},E(h)\vert_{\mH_X}) \\
=&A_X(\xi E(f)\vert_{\mH_X},\xi E(h)\vert_{\mH_X})+A_X((1-\xi)E(f)\vert_{\mH_X},\xi E(h)\vert_{\mH_X})\\
&+A_X(\xi E(f)\vert_{\mH_X},(1-\xi)E(h)\vert_{\mH_X})+A_X((1-\xi)E(f)\vert_{\mH_X},(1-\xi) E(h)\vert_{\mH_X})\\
=&\left\langle F(\overline{\xi E(f)\vert_{\mH_X}}),F(\xi E(h)\vert_{\mH_X})\right\rangle_{L^2}+\left\langle F(\overline{(1-\xi) E(f)\vert_{\mH_X}}),F(\xi E(h)\vert_{\mH_X})\right\rangle_{L^2}\\
&+\left\langle F(\overline{\xi E(f)\vert_{\mH_X}}),F((1-\xi) E(h)\vert_{\mH_X})\right\rangle_{L^2}\\\nonumber
&+\left\langle F(\overline{(1-\xi) E(f)\vert_{\mH_X}}),F((1-\xi) E(h)\vert_{\mH_X})\right\rangle_{L^2}\\
=&\left\langle K_X(\overline{f}),K_X(h)\right\rangle_{L^2}\, ,
\end{align*}
where we used $L^2$ as a short-hand notation for $L^2(\bR_+\times\bS^2;\nu_X(\eta)\td\eta\td^2\Omega_X)$. In the last step, we used that $\xi$ is real and that $\overline{E(f)}=E(\bar f)$. Combining the results for the two horizons gives the desired identity.
\end{proof}

As an immediate consequence of this result, the two-point function satisfies positivity. It remains to show the commutator property.

\subsection{The commutator property}
In this section, we show

\begin{prop}
For $0<a\ll 1$ or $\lambda\ll 1/27$ and $0<a<1$, $w(f,h)$ satisfies the commutator property, i.e.
\begin{align}
w(f,h)-w(h,f)=iE(f,h)\quad \forall f,h\in C_0^\infty(M)\, .
\end{align}
\end{prop}

\begin{proof}
The proof follows closely that of \cite[Thm. 2.1]{Dappiaggi:2009}: first, we notice that by using the identity
$\Im (x-i0^+)^{-2}=-\pi\delta^{(1)}(x)$ and partial integration, one finds \cite{Hollands:2019}
\begin{align*}
&w(f,h)-w(h,f)\\\nonumber
&=i \frac{r_+^2+a^2}{\chi}\int\left[E(f)\vert_{\mH} \partial_{U_+} E(h)\vert_{\mH}-E(h)\vert_{\mH} \partial_{U_+} E(f)\vert_{\mH}\right](U_+,\Omega_+)\td U_+\td^2\Omega_+\\\nonumber
&+i \frac{r_c^2+a^2}{\chi}\int\left[E(f)\vert_{\mHc} \partial_{V_c} E(h)\vert_{\mHc}-E(h)\vert_{\mHc} \partial_{V_c} E(f)\vert_{\mHc}\right](V_c,\Omega_c)\td V_c\td^2\Omega_c\, .
\end{align*}

Second, by \eqref{eq:SymplMorph}, $E(f,h)=\sigma(E(f),E(h))$ and $\sigma$ is independent of the Cauchy surface $\Sigma$ it is computed on. So let us take as a Cauchy surface \cite{Dappiaggi:2009} $\mH^L\cup \mathcal{B}_+\cup \Sigma_{t_0}\cup\mathcal{B}_c\cup \mHc^R$, where\\ $\Sigma_{t_0}=\{t=t_0, r_+<r<r_c\}$.
Defining
\begin{align}
J_a[f,h]=E(f)\nabla_aE(h)-E(h)\nabla_aE(f)\, ,
\end{align}
one can then write
\begin{align*}
E(f,h)=\sigma(E(f),E(h))
=&\frac{r_+^2+a^2}{\chi}\int\limits_{\mH^L}J_{U_+}[f,h]\vert_{\mH}\td U_+\td^2\Omega_+\\\nonumber
&+\int\limits_{\Sigma_{t_0}}J_a[f,h]n^a\td vol_\gamma + \frac{r_c^2+a^2}{\chi}\int\limits_{\mHc^R}J_{V_c}[f,h]\vert_{\mHc}\td V_c\td^2\Omega_c\, .
\end{align*}
We then focus on the integral over $\Sigma_{t_0}$, and aim to take the limit $t_0\to -\infty$. 

As a first step, we will split the integral further into integrals over
\begin{align*}
\Sigma_+=\Sigma_{t_0}\cap\{r_+<r\leq r_++\delta^\prime\}\, , \; \Sigma_c=\Sigma_{t_0}\cap\{r_c-\delta^\prime\leq r<r_c\}\text{ and }\Sigma_0=\Sigma_{t_0}\backslash(\Sigma_+\cup\Sigma_c)\, , 
\end{align*}
 where $\delta^\prime$ is the same constant that appears in the estimate \eqref{eq:estt}.

Focussing first on $\Sigma_0$, we use Boyer-Lindquist coordinates to compute the normal vector
\begin{align*}
n^a_{\Sigma_0}=(g^{tt})^{\frac{1}{2}}\left((\partial_t)^a-\frac{g_{t\varphi}}{g_{\varphi\varphi}}(\partial_\varphi)^a\right)
\end{align*}
and the determinant $\vert\gamma\vert=\vert g_{rr}g_{\theta\theta}g_{\varphi\varphi}\vert$. Here we denote by $g_{\mu\nu}$ the elements of the metric in the Boyer-Lindquist coordinates, and by $g^{\mu\nu}$ the element of the inverse metric, which can be found in \cite{Salazar:2017}. Written out explicitly, the integral then reads
\begin{align}
\int\limits_{\Sigma_0} J_a[f,h] n^a_{\Sigma_0}\td vol \gamma = \int\limits_{r_++\delta^\prime}^{r_c-\delta^\prime} \int\limits_{\bS^2}&\left[ \left( \frac{(r^2+a^2)^2}{\Delta_r} -\frac{a^2\sin^2\theta}{\Delta_\theta} \right) J_t[f,h]\right.\\*\nonumber
 &+ \left. a \left( \frac{r^2+a^2}{\Delta_r} - \frac{1}{\Delta_\theta} \right) J_\varphi[f,h] \right] \td^2\Omega\td r
\end{align}
The integration region is compact, and the factors appearing in front of $J_a[f,h]$ can be bounded by constants of order $(\delta^\prime)^{-1}$. Together with the estimate \eqref{eq:estt}, this means
\begin{align}
\left\vert\, \int\limits_{\Sigma_0} J_a[f,h] n^a\td vol \gamma \,\right\vert\leq C(\delta^\prime,f,h)e^{2\alpha t_0}\, ,
\end{align}
and the contribution of this term vanishes as $t_0\to -\infty$.

Next, let us analyse the integral over $\Sigma_+$. The integral over $\Sigma_c$ then works analogously. To do so, we change to Kruskal-type coordinates and note that 
\begin{align*}
\Sigma_+=\{V_+=-e^{-2\kappa_+t_0}U_+\}\cap \{U_+(t_0,r_++\delta^\prime)\leq U_+\leq 0\}\, .
\end{align*}
 Starting from a fixed $t_0$, we choose a $U_1< 0$ such that $U_+(t,r_++\delta^\prime)\leq U_1\forall t\leq t_0$. 
We then note that for any $t_0$, $\Sigma_+\cap\{U_1\leq U_+\}$ is one smooth piece of the boundary of a compact region in $\tilde M$, with the two other pieces given as $\mH^-\cap\{U_1\leq U_+\}$ and
\begin{align*}
 S_{t_0}\equiv \{0\leq V_+\leq e^{2\kappa_+t_0}U_1\}\cap \{U_+=U_1\}\, .
\end{align*} 
Since $E(f)$ and $E(h)$ satisfy \eqref{eq:KGE} on $\tilde M$, $J_a[f,h]$ is a conserved current and 
we find by an application of Stoke's theorem
\begin{align}
\int\limits_{\mathclap{\Sigma_+\cap\{U_1\leq U_+\}}} J_a[f,h]n^a\td vol\gamma=\int\limits_{\mathclap{\mH^-\cap\{U_1\leq U_+\}}}J_a[f,h]n^a\td vol_\gamma+\int\limits_{S_{t_0}}J_a[f,h]n^a\td vol_\gamma\\\nonumber
=\frac{r_+^2+a^2}{\chi}\int\limits_{\mathclap{[U_1,0]\times\bS^2}}J_{U_+}[f,h]\vert_{\mH}\td U_+\td^2\Omega_+ + \int\limits_{S_{t_0}}J_a[f,h]n^a\td vol_\gamma\, .
\end{align}
The surface $S_{t_0}$, given in Kruskal coordinates, corresponds to the interval $[0,e^{2\kappa_+t_0}U_1]$ times $\bS^2$. Since $J_a[f,h]n^a$ is a smooth function on $S_{t_0}$, the second term vanishes as $t_0\to -\infty$ and one finds
\begin{align}
\lim\limits_{t_0\to-\infty}\int\limits_{\mathclap{\Sigma_+\cap\{U_1\leq U_+\}}} J_a[f,h]n^a\td vol\gamma=\frac{r_+^2+a^2}{\chi}\int\limits_{[U_1,0]\times\bS^2}J_{U_+}[f,h]\vert_{\mH}\td U_+\td^2\Omega_+\, ,
\end{align}
compare also the proof of \cite[Thm. 2.1]{Dappiaggi:2009}.

The integral over $\Sigma_+\backslash \{U_1 \leq U_+ \}$ will now be performed in the variables $(u_+,v_+,\theta,\varphi_+)$. In these coordinates,
\begin{align*}
\Sigma_+\backslash \{U_1\leq U_+\}=\{v_+=2t_0-u_+\}\cap\{u_+(t_0,r_++\delta^\prime)\leq u_+\leq u_1\}\, ,
\end{align*} 
where $u_1=-\kappa_+^{-1}\ln \vert U_1\vert$. With this, one can derive by a direct computation the explicit formula
\begin{align}
\label{eq:langesIntegral}
\int\limits_{\mathclap{\Sigma_+\backslash\{U_1\leq U_+\}}} J_a[f,h]n^a \td vol\gamma &= \int\limits_{\mathclap{(-\infty,u_1)\times \bS^2}} \mathbb{1}_{(u_+(t_0,r_++\delta^\prime),\infty)}\left(\left[\frac{r^2+a^2}{\chi}-\frac{a^2\sin^2\theta\Delta_r}{\chi(r^2+a^2)\Delta_\theta}\right]\right. \\\nonumber
&\times\left.\left.\vphantom{\left[\frac{r^2}{\chi}\right]}\left[J_{u_+}[f,h]+J_{v_+}[f,h]\right]+G_+ J_{\varphi_+}[f,h] \right)\right\vert_{\mathrlap{v_+=2t_0-u_+}}\td u_+\td \Omega_+\, .
\end{align}
Here, 
\begin{align*}
G_+=\frac{a(r_+^2-r^2)}{\chi(r_+^2+a^2)}-\frac{a\rho_+^2\Delta_r}{\chi(r^2+a^2)(r_+^2+a^2)\Delta_\theta}
\end{align*}
is a function of $r$ and $\theta$ that vanishes as $r\to r_+$. Due to the estimate \eqref{eq:estU}, the currents $J_{u_+}[f,h]$, $J_{v_+}[f,h]$ and $J_{\varphi_+}[f,h]$ can be bounded by $C e^{2\alpha u_+}$. Note that by the construction of $\Sigma_+$, $r(u_+)$ ranges at most from $r_+$ to $r_++\delta^\prime$, independent of $t_0$. Thus, the additional factors in the integral can be estimates by 
\begin{align*}
\left\vert\frac{r^2+a^2}{\chi}-\frac{a^2\sin^2\theta\Delta_r}{\chi(r^2+a^2)\Delta_\theta}\right\vert&\leq \frac{(r_++\delta^\prime)^2+a^2}{\chi}+\frac{a^2\Delta_r\vert_{r_++\delta^\prime}}{\chi(r_+^2+a^2)}\\
\vert G_+(r,\theta)\vert&\leq \frac{\vert a\vert ((r_++\delta^\prime)^2-r_+^2)}{\chi(r_+^2+a^2)}+\frac{\vert a\vert\Delta_r\vert_{r_++\delta^\prime}}{\chi(r_+^2+a^2)}
\end{align*} 
while the cutoff function can be estimated by $1$. Combining these estimates, the integrand on the right hand side can be bounded by $C(f,h,\delta^\prime)e^{2\alpha u_+}$, which is independent of $t_0$ and integrable on $(-\infty,u_1)\times \bS^2$. By the dominated convergence theorem, we may thus take the limit $t_0\to -\infty$ under the integral sign. This corresponds to taking $r_*=t_0-u_+\to -\infty$ and hence $r$ to $r_+$. One finds
\begin{align*}
&\lim\limits_{t_0\to-\infty}\int\limits_{\mathclap{\Sigma_+\backslash\{U_1\leq U_+\}}} J_a[f,h]n^a \td vol\gamma = \int\limits_{\mathclap{(-\infty,u_1)\times \bS^2}} \left.\left(J_{u_+}[f,h]+J_{v_+}[f,h]\right)\right\vert_{v_+\to-\infty}
\frac{r_+^2+a^2}{\chi}\td u_+\td \Omega_+\, .
\end{align*}
Changing back to Kruskal-type coordinates one notes that 
\begin{align*}
J_{v_+}[f,h]=\kappa_+V_+J_{V_+}[f,h]\, ,
\end{align*}
which thus vanishes in the limit $t_0\to -\infty$. The final result is
\begin{align}
\lim\limits_{t_0\to-\infty}&\int\limits_{\mathclap{\Sigma_+\backslash\{U_1\leq U_+\}}} J_a[f,h]n^a \td vol\gamma =\frac{r_+^2+a^2}{\chi}\int\limits_{\mathclap{(-\infty,U_1)\times\bS^2}}J_{U_+}[f,h]\vert_{\mH}\td U_+\td^2\Omega_+\, .
\end{align}
Collecting the pieces, we thus find 
\begin{align}
\lim\limits_{t_0\to-\infty}&\int\limits_{\Sigma_+} J_a[f,h]n^a \td vol\gamma =\frac{r_+^2+a^2}{\chi}\int\limits_{\mH^-}J_{U_+}[f,h]\vert_{\mH}\td U_+\td^2\Omega_+\, .
\end{align}
This finishes the proof.
\end{proof}

Summarizing, we have shown in this section that $w(f,h)$ indeed defines the two-point function for some state on $\mathcal{A}$.

\section{The Hadamard property}
\label{sec:Had}
Finally, we would like to demonstrate that the quasi-free state defined by the two-point function $w(f,h)$ is a Hadamard state. The strategy of the proof will be as follows: To show that the condition on the wavefront set of $w$ is satisfied, we start by demonstrating that instead of considering all points in $T^*(M\times M)$, it is sufficient to focus on the primed wavefront set restricted to the diagonal 
\begin{align}
\Delta_{T^*(M\times M)}=\{(x,k;x,k)\in T^*(M\times M)\}\, ,
\end{align}
and in addition it suffices to consider $k$ null. Moreover, instead of considering points $(x,k)\in T^*M$, one can work with bicharacteristic strips
\begin{align}
 B(x,k)&=\left\{(x^\prime,k^\prime)\in T^*M:(x^\prime,k^\prime)\sim (x, k)\right\}\, ,\quad B(x,0)=\{(x,0)\} \,.
 \end{align}

After this, we proceed similar to \cite{Dappiaggi:2009} and show the Hadamard property in a subregion $\mathcal{O}$ of $M$ first. As in \cite{Dappiaggi:2009}, in this subregion, the Hadamard property follows from a slight modification of the proof of \cite[Thm. 5.1]{Sahlmann:2000}. However, in contrast to \cite{Dappiaggi:2009}, one cannot take this subregion to be the whole region $\rI$. Nonetheless, by Lemma~\ref{lem:region}, $\mathcal{O}$ will still be sufficiently large to cover all cases $B(x,k)$, where the null geodesic corresponding to the projection of $B(x,k)$ to the manifold, which we will call bicharacteristic and denote $B_M(x,k)$, does not end at one of the horizons $\mH$ or $\mHc$.

It then remains to consider cases $B(x,k)$ where the corresponding null geodesic ends at one of the horizons. To handle these cases, we will use a number of cutoff-functions to split the two-point function into a piece whose wavefront set may contain $(x,k;x,-k)$ and a remainder. We will compute the wavefront set of the first piece explicitly, and then show that $(x,k;x,-k)$ is a direction of rapid decrease for the remainder. 

This part of the proof is similar in idea to the proof in \cite{Hollands:2000, Gerard:2014}, though some aspects of it are more complicated. For example, the splitting of the two-point function depends on both $x$ and $k$. This idea was also applied in \cite{Hollands:2019}, but was not made very explicit in that paper.

Before we start the proof, let us  define the forward and backward null cones
\begin{align}
\mathcal{N}^\pm&=\{(x,k)\in T^*M\backslash o:g^{-1}(x)(k,k)=0, \pm k \text { future-directed}\}\, .
\end{align} 

Then, as a first step, we note that by the Propagation of Singularities theorem \cite[Lemma 6.5.5]{Duistermaat:1972}, if $(x_1,k_1;x_2,k_2)\in T^*(M\times M)\backslash o$ is in $\WF^\prime (w)$, then $k_1$ and $k_2$ are null covectors (or zero) and $B(x_1,k_1)\times B(x_2,k_2)\subset\WF^\prime(w)$. Hence, instead of considering points $(x,k;y,l)\in T^*(M\times M)$, we can consider any pair of bicharacteristic strips $B(x,k)\times B(y,l)$, freely choosing any representative.

Additionally, as a consequence of \cite[Thm.6.5.3]{Duistermaat:1972} and as noted in the proof thereof, denoting the integral kernel of $E$ also by $E$,
\begin{align}
\label{eq:WF(E)}
\WF^\prime (E)=\mC^+\cup \mC^-\, .
\end{align}

Next, let us demonstrate the following Lemma, which is related to \cite[Prop. 6.1]{Strohmaier:2002}:
\begin{lem}
\label{lem:Diag}
If the two-point function $w$ satisfies
\begin{align}
\label{eq:altWFcon}
\WF^\prime (w)\cap \Delta_{T^*(M\times M)}\subset \mN^+\times\mN^+\, ,
\end{align}
 where $\Delta_{T^*(M\times M)}$ is the diagonal in $T^*(M\times M)$, then the corresponding quasi-free state on $\mathcal{A}$ has the Hadamard property.
 \end{lem}

\begin{proof}
 Assume \eqref{eq:altWFcon} holds. Then, if $(x,k;x,k)\in \WF^\prime(w)$, $(x,-k;x,-k)$ cannot be in $\WF^\prime(w)$.

Let us first assume $B_M(x_0,k_x)\neq B_M(y_0,k_y)$ for $k_x$, $k_y$  both non-zero, or one of them, say $k_x$, is zero but $B_M(y_0,k_y)$ does not intersect $\{x_0\}$. Then we can find some spacelike Cauchy surface $\Sigma$ of $M$, which is intersected by the corresponding bicharacteristics in two distinct points $x_1$ and $y_1$. Let $f,h\in C_0^\infty(M;\bR)$ be supported in spacelike separated neighbourhoods of $x_1$ and $y_1$. 
By Proposition~\ref{prop:K_X}, we can write 
\begin{align*}
w(f,h)= \sum\limits_X \left\langle K_X(\overline{f}),K_X(h)\right\rangle_{L^2(\bR_+\times\bS^2;\nu_X(\eta)\td\eta\td^2\Omega_X)}\,.
\end{align*}
After fixing some coordinates on a neighbourhood of $x_1$ and $y_1$, we write $f_{k}(x)=(2\pi)^{-2}e^{i k\cdot x} f(x)$, where $k\cdot x$ should be understood as the usual product in $\bR^4$. We can then use the Cauchy-Schwarz inequality to deduce that \cite{Hollands:2019}
\begin{align*}
\vert w(f_k, h_l)\vert^2\leq \vert w(f_k,f_{-k})\vert\vert w(h_{-l},h_l)\vert\, .
\end{align*}
Since the supports of $f$ and $h$ are spacelike separated, we have by the commutator property at spacelike separation
\begin{align*}
\vert w(f_{k}, h_{l})\vert^2=\vert w(h_{l}, f_{k})\vert^2\leq \vert w(f_{-k},f_{k})\vert\vert w(h_{l},h_{-l})\vert\, .
\end{align*}
If $\WF^\prime(w)\cap \Delta_{T^*(M\times M)}\subset \mN^+\times\mN^+$, then one can choose $f$ and $h$ so that at least one of the two estimates for $\vert w(f_k, h_l)\vert^2$ is rapidly decreasing in $(k,l)$ in a small conic neighbourhood of $(k_x,k_y)$ parallel transported to $(x_1,y_1)$, and hence such points are not in the wavefront set of $w$.

In the case where $k_x=0$, but $B_M(y_0,k_y)$ contains $x_0$, the above argument does not hold, since no points of the bicharacteristics of $(x_0,0)$ and $(y_0,k_y)$ are spacelike separated. However, we may use that $\WF(E)$, which by the commutator property is the same as $WF(w-\tilde w)$, $\tilde w (f,h)=w(h,f)$, does not contain points of the form $(x,0;y,k)$. Hence, if such a point is in $\WF(w)$, it must also be in $\WF(\tilde w)$, so that the two singular contributions can cancel out. This entails that $(y,k;x,0)$ must be in $\WF(w)$ if $(x,0;y,k)$ is \cite{Dappiaggi:2009}. Let $f,h\in C_0^\infty(M;\bR)$ supported in small neighbourhoods of $x_0$. Then if
\begin{align*}
\vert w(f, h_{k})\vert^2\leq \vert w(f,f)\vert\vert w(h_{-k},h_{k})\vert
\end{align*}
is not rapidly decreasing in $\vert k\vert$, then
\begin{align*}
\vert w(h_k, f)\vert^2\leq \vert w(f,f)\vert\vert w(h_{k},h_{-k})\vert\, 
\end{align*}
must not be rapidly decreasing in $\vert k\vert$ either.
Again, if $\WF^\prime(w)\cap \Delta_{T^*(M\times M)}\subset \mN^+\times\mN^+$, one can find $f$, $h$ so that at least one of the two estimates decreases rapidly, excluding  any points of the form $(x,0;y,k)$ from $\WF(w)$.

It remains to consider bicharacteristics with identical projections to the manifold. They can be represented by points in $T^*(M\times M)$ of the form $(x_0,k_x;x_0,b k_x)$ with $b\neq 0$, $b\in\bR$. Then, as above, taking some $f,h\in C_0^\infty(M;\bR)$ supported in small neighbourhoods of $x_0$, we get
\begin{align*}
\vert w(f_k,h_{l})\vert^2\leq \vert w(f_k,f_{-k})\vert\vert w(h_{-l},h_{l})\vert\, .
\end{align*}
By our assumption, this is rapidly decreasing for some choice of $f$ and $h$ in a small conic neighbourhood of $(k_x,b k_x)$ unless $(x_0,k_x)\in \mN^+$ and $(x_0,b k_x)\in\mN^-$, or in other words $(x_0,k_x)\in \mN^+$ and $b<0$. 
Combining this with the other cases above, we have shown that our assumption entails $\WF^\prime (w)\subset \mN^+\times\mN^+$. This also means that $\WF^\prime (\tilde w)\subset \mN^-\times\mN^-$, so the two wavefront sets do not overlap. Then, similar to the case above, we can use that $E=w-\tilde w$ and thus
\begin{align*}
\mC^+\cup\mC^-=\WF^\prime (E)=\WF^\prime(w-\tilde w) = \WF^\prime(w)\cup\WF^\prime(\tilde w)\subset \mN^+\cup\mN^-\, ,
\end{align*}
where the third equal sign is due to the fact that the wavefront sets of $w$ and $\tilde w$ do not overlap, see also the proof of \cite[Prop. 6.1]{Strohmaier:2002}. $\WF^\prime(w)$ only intersects $\mC^+$, while $\WF^\prime(\tilde w)$ only intersects $\mC^-$. So the above equation can only be satisfied if
\begin{align*}
\WF^\prime(w)=\mC^+\, .
\end{align*}

Hence, we see that it is actually sufficient to show that $\WF^\prime(w)\cap \Delta_{T^*(M\times M)}$ is contained in $\mN^+\times\mN^+$ \cite{Gerard:2020}.  
\end{proof}

After this consideration, we now want to prove the Hadamard condition in a subregion $\mathcal{O}$ of $M$. We choose $\mathcal{O}\subset \rI$ to be the open region in $\rI$ where the Killing vector fields $\partial_{t_X}$
are timelike for both $X=+$ and $X=c$. As was demonstrated in Lemma~\ref{lem:region}, any inextendible null geodesic not ending at either $\mH$ or $\mHc$ in the past must pass through $\mathcal{O}$ as long as $a$, $\lambda$ are sufficiently small.  Applying also the Propagation of Singularity theorem, we can even consider all points in the set
\begin{align}
B(\mathcal{O})\times B(\mathcal{O})=\{(x_1,k_1;x_2,k_2)\in T^*(M\times M)\backslash o: B_M(x_i,k_i)\cap\mathcal{O}\neq\emptyset\, , \, \, i=1,2\}\, .
\end{align}  

Thus, our goal is to show

\begin{prop}
\label{prop:HadO}
Let $w(f,h)$ as defined in \eqref{eq:2ptfct}. Then 
\begin{align}
\label{eq:HadCondO}
\WF^\prime (w)\cap(B(\mathcal{O})\times B(\mathcal{O}))=\mC^+\cap(B(\mathcal{O})\times B(\mathcal{O}))\, .
\end{align}
\end{prop}

\begin{proof}
In the following, we will show that
\begin{align}
\label{eq:ClaimO}
\WF^\prime(w)\cap T^*_{(x_0,x_0)}(M\times M)\cap \Delta_{T^*(M\times M)}\subset \mN^+\times\mN^+\, 
\end{align}
for any $x_0\in \mathcal{O}$.
By Lemma \ref{lem:Diag}, and the Propagation of Singularities theorem, this result will then imply \eqref{eq:HadCondO}.

The proof, similar to the one of \cite[Prop. 4.3]{Dappiaggi:2009}, follows largely part 1) and 2) of the proof of \cite[Thm. 5.1]{Sahlmann:2000}, which is based on the characterization of the wavefront set by \cite[Prop. 2.1]{Verch:1998}. Parts 3)-6) of the proof of \cite[Thm. 5.1]{Sahlmann:2000} are already covered by Lemma \ref{lem:Diag}.

The first step in our adaptation of the proof is to show that the pieces $w_X$ of our two-point function are "KMS like" \cite{Dappiaggi:2009} at inverse temperature $2\pi\kappa_X^{-1}$ with respect to the isometries induced by $\partial_{t_X}$, which is weaker than the passivity condition used in \cite{Sahlmann:2000}, but still sufficient:
\begin{lem}
\label{lem:KMSlike}
Denote by $\phi^X_b$ the flow generated by the Killing field $\partial_{t_X}$ acting on $f\in C_0^\infty(\rI)$ by $\phi_b^Xf(u_X,v_X,\theta,\varphi_X)=f(u_X-b,v_X-b,\theta,\varphi_X)$. Then
\begin{align}
K_X^\rI(\phi_b^Xf)(\omega,\theta,\varphi_X)=e^{ib\omega}K_X^\rI(f)(\omega,\theta,\varphi_X)\, .
\end{align}
Moreover, $w_X$ is "KMS like" \cite{Dappiaggi:2009} at inverse temperature $2\pi\kappa_X^{-1}$, i.e. for any $h\in C_0^\infty(\bR;\bR)$ and any pair $f_{1/2}\in C_0^\infty(\rI;\bR)$ of real test functions
\begin{align}
\int\limits_\bR \hat h(t)\left\langle K_X^\rI(f_1),K_X^\rI(\phi_t^X f_2)\right\rangle_{L^2}\td t=\int\limits_\bR \hat h\left(t+\frac{2\pi i}{\kappa_X}\right)\left\langle K_X^\rI(\phi_t^X f_2),K_X^\rI(f_1)\right\rangle_{L^2}\td t
\end{align}
\end{lem}
\begin{proof}
Since $\partial_{t_X}$ are Killing fields, the commutator function $E$ satisfies 
\begin{align*}
E(\phi_b^X f)(u_X,v_X,\theta,\varphi_X)=E(f)(u_X-b,v_X-b,\theta,\varphi_X)\,.
\end{align*}
The first claim thus follows immediately from the definition of $K_X^\rI$. Next, let $h \in C^\infty_0(\bR;\bR)$ and $f_{1/2}\in C_0^\infty(\rI;\bR)$. Then
\begin{align}
&\int\limits_\bR \hat h(t)\left\langle K_X^\rI(f_1),K_X^\rI(\phi_t^X f_2)\right\rangle_{L^2}\td t\\\nonumber
&=\int\limits_\bR \hat h(t) \int\limits_{\bR\times\bS^2}\mu_X(\omega)\overline{K_X^\rI(f_1)}(\omega,\theta,\varphi_X)e^{i\omega t}K_X^\rI(f_2)(\omega,\theta,\varphi_X)\td \omega\td \Omega_X\td t\, .
\end{align}
By the definition of $K_X^\rI$, if $f$ is a real function, then $\overline{K_X^\rI(f)}(\omega, \theta,\varphi_X)=K_X^\rI(f)(-\omega,\theta,\varphi_X)$. Combining this with the fact that $\mu_X(\omega)=e^{2\pi\omega/\kappa_X}\mu_X(-\omega)$, we can write the above as
\begin{align*}
&\int\limits_\bR \hat h(t)\int\limits_{\bR\times\bS^2}\mu(\tilde\omega) e^{-2\pi\tfrac{\tilde\omega}{\kappa_X}} e^{-i\tilde\omega t} K_X^\rI(f_1)(\tilde\omega,\theta,\varphi_X) K_X^\rI(f_2)(-\tilde\omega,\theta,\varphi_X)\td\tilde\omega\td^2\Omega_X\td t\\
=&\int\limits_{\bR-i\tfrac{2\pi}{\kappa_X}}\hat h\left(t+i\frac{2\pi}{\kappa_X}\right)\int\limits_{\bR\times\bS^2}\mu_X(\tilde\omega)e^{-i\tilde\omega t}\overline{K_X^\rI(f_2)}(\tilde\omega,\theta,\varphi_X)K_X^\rI(f_1)(\tilde\omega,\theta,\varphi_X)\td\tilde\omega\td^2\Omega_X\td t\\
=&\int\limits_\bR \hat h \left(t+i\frac{2\pi}{\kappa_X}\right)\left\langle K_X^I(\phi_t^X f_2),K_X^\rI(f_1)\right\rangle_{L^2} \td t\,.
\end{align*}
The last step works since $\hat h(t)$, the Fourier transform of a compactly supported function, is entire analytic and vanishes for $\Re t\to\pm\infty$ as long as $\Im t$ remains finite. In addition, for any $f_{1/2}\in C_0^\infty(\rI;\bR)$, the function 
\begin{align*}
t\mapsto& \left\langle K_X^\rI(f_1), K_X^\rI(\phi_t^X f_2)\right\rangle_{L^2(\bR\times\bS^2;\mu_X(\omega)\td\omega\td^2\Omega_X)}\\
&=\left\langle K_X^\rI(f_1), e^{i\omega t}K_X^\rI(f_2)\right\rangle_{L^2(\bR\times\bS^2;\mu_X(\omega)\td\omega\td^2\Omega_X)}
\end{align*}
has an analytic continuation to $\Im t\in [0,2\pi/\kappa_X)$. This can be seen by an explicit calculation using the form of $\mu(\omega)$, as well as the estimate \eqref{eq:estU} or \eqref{eq:estV}.
\end{proof}

This Lemma can be applied to show the parts of \cite[Prop. 2.1]{Sahlmann:2000} which are relevant for the present case, by following the proof of \cite[Prop. 2.1]{Sahlmann:2000} step by step\footnote{Note that the convention for the Fourier transform in \cite{Verch:1998,Sahlmann:2000} differs from ours.}:
\begin{lem}
\label{lem:SV,Prop2.1}
 For any $(g_1^\lambda)_{\lambda>0}$, $(g_2^\lambda)_{\lambda>0}\subset C_0^\infty(\mathcal{O};\bR)$ such that $w_X(g_i^\lambda,g_i^\lambda)\leq c\left(1+\lambda^{-1}\right)^{s}$ for some $c>0$ and $s>0$ , there exist $h\in C_0^\infty(\bR^2):\hat h(0)=1$ and for any $(k_0,k_0^\prime)\in \bR^2\backslash\{0\}$, such that $k_0^\prime>0$, there is an open neighbourhood $V_\epsilon$ in $\bR^2\backslash\{0\}$ of $(k_0,k_0^\prime)$ such that $k_2>\epsilon>0\, \forall (k_1,k_2)\in V_\epsilon$  and such that $\forall N\in\bN$ $\exists C_N>0$, $\lambda_N>0$:
 \begin{align}
 \sup\limits_{k\in V_\epsilon}\left\vert \int e^{i\lambda^{-1}k\cdot t}\hat h(t)w_X(\phi^X_{t_1}g_1^\lambda,\phi^X_{t_2}g_2^\lambda)\td^2 t\right\vert <C_N\lambda^N\quad \forall 0<\lambda<\lambda_N\, 
 \end{align}
 \end{lem}
As mentioned in the proof of \cite[Prop. 2.1]{Sahlmann:2000}, by an application of \cite[Lemma 2.2 b)]{Verch:1998}, this continues to hold when $\hat h$ is replaced by $\phi\cdot\hat h$ for some $\phi\in C_0^\infty(\bR^2)$ after potentially shrinking $V_\epsilon$. It also continues to hold if the functions $g_i^\lambda$ depend on additional parameters, see \cite[Rem. 2.2]{Sahlmann:2000}.

Following the proof of \cite[Thm. 5.1]{Sahlmann:2000}, let us now consider any fixed point $x_0\in \mathcal{O}$. In a neighbourhood $\mathcal{U}_{x_0}$ of $x_0$, we then define a coordinate chart 
\begin{align*}
\psi:\mathcal{U}_{x_0}&\to\psi(\mathcal{U}_{x_0})\subset\bR^4\, , & x&\to (t_X(x)=\tfrac{1}{2}(u_X+v_X)(x)-t_{X,0}, \vec{x}(x))
\end{align*}
where $t_{X,0}$ and $\vec{x}$ are chosen such that $\psi(x_0)=0$, for example by taking $\vec{x}(x)$ to be the Cartesian coordinates corresponding to $(r(x), \theta(x), \varphi_X(x))$, and then shifting the origin of the coordinates to $\vec{x}(x_0)$.

The coordinate chart should be built in such a way that there is some constant $c>0$ so that for $\vert t\vert<c$, the diffeomorphism $\rho^X_t$ induced by Killing vector field $\partial_{t_X}$ can be written as 
\begin{align*}
\psi\circ\rho^X_t(x)=(t_X(x)+t,\vec{x}(x))\, 
\end{align*}
on a sufficiently small neighbourhood $K\subset \mathcal{U}_{x_0}$ of $x_0$.
In addition, spatial translations $\rho_{\vec{y}}$, for $\vec{y}$ in a sufficiently small neighbourhood $B$ of $0$ in $\bR^{3}$, are defined  on $K$ by 
\begin{align*}
\rho_{\vec{y}}(x)&=\psi^{-1}\circ \tilde{\rho}_{\vec{y}}\circ\psi(x) & \tilde{\rho}_{\vec{y}}(t_X,\vec{x})=(t_X,\vec{x}+\vec{y})\, 
\end{align*}
The corresponding pullbacks acting on functions on $K$ are then written as $\phi^X_t$, which was already used earlier, and $R_{\vec{y}}$.

By our assumption, $\partial_{t_X}$ is timelike and future-pointing on $\mathcal{O}$. As a consequence, a null-covector $(x,k)\in T^*_{\mathcal{O}}M$ is future-pointing iff $k^0=\langle k,\partial_{t_X} \rangle >0$.

After the construction of the coordinate chart, let us consider $(x_0,k_x)\in \mN^-$, with $x_0$ in $\mathcal{O}$. We will use the description of the wavefront set given in Proposition~\ref{prop:V,Prop2.1}, see also \cite[Lemma 3.1]{Sahlmann:2000}, to show that $(x_0,k_x; x_0,-k_x)$ is not in $\WF(w)$. 
 
To this end, let us define $H\in C_0^\infty(\psi(\mathcal{U}_{x_0})\times \psi(\mathcal{U}_{x_0}))$ as $H(t_X,\vec{x},t_X^\prime,\vec{x}^\prime)=\phi(t_X,t_X^\prime)\hat h(t_X,t_X^\prime)\zeta(\vec x,\vec{x}^\prime)$, where we take the $h\in C_0^\infty(\bR^2)$ as in Lemma \ref{lem:SV,Prop2.1}, and $\phi\in C_0^\infty((-c,c);\bR)$, $\zeta\in C_0^\infty(B\times B; \bR)$ such that $H(0)=1$.
 
We identify $T^*_{x_0}M$ with $\bR^4$ in our coordinate chart $\psi$. Then, we take $V\subset(\bR^4\times \bR^4)\backslash \{0\}$ to be an open neighbourhood  of $(k_x,-k_x)$ such that for all $(k,k^\prime)\in V$, $(k^0,k^{\prime 0})=(\langle k,\partial_{t_X}\rangle,\langle k^\prime,\partial_{t_X}\rangle)\in V_\epsilon$, with $V_\epsilon$ as in Lemma~\ref{lem:SV,Prop2.1} for some $\epsilon>0$.

In addition, let us note that functions of the form as in Proposition~\ref{prop:V,Prop2.1} satisfy the condition of Lemma~\ref{lem:SV,Prop2.1}: 
 let $g_i\in C_0^\infty(\psi(\mathcal{U}_{x_0});\bR)$ with support in a sufficiently small neighbourhood $\psi(K)$ of $0$. Let us also assume $\widehat{ g_1\otimes g_2}(0,0)=1$. Then, taking any $p\geq 1$ and $\lambda\leq 1$, we set $g_i^{\lambda}(x)=g_i(\lambda^{-p}(\psi(x)))$ for $x\in \mathcal{U}_{x_0}$ and $g_i^{\lambda}(x)=0$ outside of $\mathcal{U}_{x_0}$. For $\lambda >1$, set $g_i^\lambda(x)=g_i^1(x)$. Then $\supp(g_i^\lambda)\subset K$, so that  time translations by $t\in (-c,c)$ and spatial translations by $\vec{y}\in B$ as defined above are well-defined for all $\lambda$. Moreover, we can use that by \eqref{eq:estiW},
\begin{align}
\vert w_X(g_i^{\lambda},g_i^{\lambda})\vert\leq C\norm{g_i^\lambda}_{C^m}^2\, ,
\end{align}
and since the functions $g_i^\lambda$ are supported away from the horizons, this norm can be taken  using partial derivatives in the $\psi$-coordinate system as the linear independent vector fields. Taking into account that $\partial_x f(\lambda^{-p}x)=\lambda^{-p}\partial_yf(y)\vert_{y=\lambda^{-p}x}$, we get for $\lambda<1$
\begin{align}
\vert w_X(g_i^{\lambda},g_i^{\lambda})\vert\leq C\norm{g_i(\lambda^{-p}x)}_{C^m}^2\leq C^\prime \lambda^{-2mp}\norm{g_i(x)}_{C^m}^2\leq C^{\prime\prime}\norm{g_i}_{C^m}^2(1+\lambda^{-1})^{2mp}\, .
\end{align}
Hence, the $g_i^\lambda$ of the form as in Proposition~\ref{prop:V,Prop2.1} satisfy the condition of Lemma~\ref{lem:SV,Prop2.1} for $c=C^{\prime\prime}\norm{g_i}_{C^m}^2$ and $s=2mp$.

Then the claim \eqref{eq:ClaimO} follows from Lemma~\ref{lem:SV,Prop2.1} and the form of the wavefront set as in Proposition~ \ref{prop:V,Prop2.1} by using the estimate
\begin{align*}
&\sup\limits_{(k,k^\prime)\in V} \left\vert \int e^{i\lambda^{-1}(k,k^\prime)\cdot(x,x^\prime)} H(x,x^\prime) w_X\left(\phi^X_t R_{\vec{x}}g_1^\lambda,\phi^X_{t^\prime}R_{\vec{x}^\prime}g_2^\lambda\right) \td^4x \td^4 x^\prime \right\vert\\
=&\sup\limits_{(k,k^\prime)\in V} \left\vert \int e^{i\lambda^{-1}(\vec k \vec x +\vec{k}^\prime\vec{x}^\prime)}\zeta(\vec x, \vec{x}^\prime)\times\right.\\
&\times\left.\left[\int e^{i\lambda^{-1}(k^0t+k^{0\prime}t^\prime)}\phi(t,t^\prime)\hat h(t,t^\prime)w_X\left(R_{\vec{x}}g_1^\lambda,\phi^X_{t^\prime-t}R_{\vec{x}^\prime}g_2^\lambda\right)\td t\td t^\prime\right] \td^3\vec{x}\td^3\vec{x}^\prime\right\vert\\
\leq&\sup\limits_{(k,k^\prime)\in V}\int \left\vert \zeta(\vec x,\vec{x}^\prime)\right\vert\times\\
&\times\left\vert \int e^{i\lambda^{-1}(k^0t+k^{0\prime}t^\prime)}\phi(t,t^\prime)\hat h(t,t^\prime)w_X\left(R_{\vec{x}}g_1^\lambda,\phi^X_{t^\prime-t}R_{\vec{x}^\prime}g_2^{\lambda}\right)\td t\td t^\prime\right\vert \td^3\vec{x}\td^3\vec{x}^\prime\\
\leq& \sup\limits_{(k,k^\prime)\in V} \int \left\vert \zeta(\vec x,\vec{x}^\prime)\right\vert C_N \lambda^{N}\td^3\vec{x}\td^3\vec{x}^\prime\\
\leq& \tilde{C}_N\lambda^{N}\quad\quad \forall 0<\lambda<\lambda_N\leq 1\, .
\end{align*}

This completes the proof of Proposition~\ref{prop:HadO}.
\end{proof}


We have thus established the Hadamard property in a subregion of $M$, which must be intersected by all null geodesics which do not end at either $\mH$ or $\mHc$. Applying Lemma~\ref{lem:Diag}, the remaining case we need to consider is 

\begin{prop}
\label{prop:HadM/O}
Let $(x_0,k_0)\in T^*M\backslash o$ such that $B_M(x_0,k_0)\cap\mathcal{O}=\emptyset$. Assume that $\lambda$, $a$ are chosen such that Lemma~\ref{lem:region} and the results of \cite{Hintz:2015} are valid. If $(x_0,k_0;x_0,k_0)$ is in $\WF^\prime(w)$, then $(x_0,k_0)\in \mN^+$.
\end{prop}

\begin{proof}
We will work in the +-Kruskal coordinates, and assume that $B_M(x_0,k_0)$ intersects $\mH$. The case where it intersects $\mHc$ works analogously. We will denote by $\psi_+:M_+\to\bR^2\times \bS^2$ the coordinate diffeomorphism of the Kruskal coordinates, and we will write $(U,\Omega, \xi,\sigma)$ for points in $T^*(\bR\times\bS^2)$, where $V_+=0$, i.e. $\Omega=(\theta,\varphi_+)$ and $\sigma \in T^*_{\Omega}(\bS^2)$. We will also identify $T_x^*M$ with $\bR^4$ in these coordinates, so that covectors at different points can be compared.

Let $K\subset M$ be a compact neighbourhood of $x_0$, and  $V$ a small conic neighbourhood of $k_0$ identified with an element of $\bR^4$ under $\psi_+$. We may choose them such that there is a compact set $\mathcal{U}\subset \mH$ such that all  $B_M(x,k)$ with $x\in K$, and $k\in V$ intersect $\mH$ in the interior of $\mathcal{U}$.

We can then find a function $h\in C_0^\infty(\bR\times\bS^2)$ such that $h=1$ on $\mathcal{U}$. Let us also define a function $\zeta\in C_0^\infty(\tilde M)$ such that $\zeta\vert_{\mH}=1$ on $\supp(h)$. Then, following the ideas in \cite{Hollands:2000,Gerard:2014}, we may consider the splitting:
\begin{align}
\label{eq:splitting}
w=&(h\cdot A_+\cdot h)\left(tr_{\mH}\circ (\zeta\cdot E),tr_{\mH}\circ(\zeta\cdot E)\right)\\\nonumber
&+  A_+\left((1-h)\cdot tr_{\mH}\circ E, h \cdot tr_{\mH}\circ E\right)\\\nonumber
&+  A_+\left(h\cdot tr_{\mH}\circ E, (1-h)\cdot tr_{\mH}\circ E\right)\\\nonumber
&+A_+\left((1-h)\cdot tr_{\mH}\circ E,((1-h)\cdot tr_{\mH}\circ E\right)\\\nonumber
&+ w_c\, .
\end{align}
Here, we have denoted the restriction to $\mH$ by $tr_{\mH}$.

As mentioned above, we will start by analysing the first piece on the right hand side of \eqref{eq:splitting}, and show that its contribution to the wavefront set satisfies \eqref{eq:altWFcon}. 

To do so, we notice that $\zeta\cdot E $, $tr_{\mH}$ and $h \cdot A_+\cdot h$ are properly supported, i.e they satisfy \cite[Eq. (8.2.13)]{Hoermander}. Thus, we may apply \cite[Thm. 8.2.14]{Hoermander} to determine the wavefront set from the wavefront sets of $A_+$, $tr_{\mH}$ and $E$. We find by direct computation (see also \cite{Hollands:2000, Hollands:2019})
\begin{align}
\WF^\prime(A_+)= \left\{ (U,\Omega, \xi,\sigma ; U^\prime,\Omega,\xi,\sigma)\in T^*(\bR\times\bS^2\times\bR\times\bS^2) \backslash o: \right. \\ \nonumber
\left. \xi>0 \text{ if } U=U^\prime\, , \xi=0 \text{ else } \right\} \, ,
\end{align}
and by an application of \cite[Thm. 8.2.4]{Hoermander}
\begin{align}
\WF^\prime (tr_{\mH})=\left\{ (U,\Omega, \xi,\sigma; x,k) \in T^*(\bR\times\bS\times\tilde M): \psi_+(x)=(U,0,\Omega),\right.\\\nonumber \left.{}^t\td \psi_+(x)(\xi,\eta,\sigma)=k\text{ for some }\eta\in\bR\right\}\, .
\end{align}

Taking into account Lemma~\ref{lem:Hgeo}, which is an analogue of \cite[Lemma 5.1]{Gerard:2014}\footnote{This result, together with \cite[Thm. 8.2.4]{Hoermander} also allows one to make sense of the map $tr_{\mH}\circ E:C_0^\infty(M)\to C^\infty(\bR\times\bS^2)$ without the cutoff function $\zeta$.}, we find by \cite[Thm. 8.2.14]{Hoermander}
\begin{align}
&\WF^\prime((h\cdot A_+\cdot h)\left(tr_{\mH}\circ (\zeta\cdot E),tr_{\mH}\circ(\zeta\cdot E)\right))\\\nonumber
\subset&\left\{\vphantom{T^*_{\mH}(\tilde M)}(x_1,k_1;x_2,k_2)\in T^*(M\times M)\backslash o: \exists(y,l)\in T^*_{\mH}(\tilde M):\right.\\\nonumber
&\left. (x_i,k_i)\sim (y,l)\, , \, \, i=1,2\,;\,{}^t\td(\psi_+^{-1})(\psi_+(y))l=(\xi,\eta,\sigma)\text{ with }\xi>0\right\}\\\nonumber
&\subset \mathcal{N}^+\times\mathcal{N}^+\, . 
\end{align}
This shows the result for the first piece.

Finally, we want to show that for the remaining terms on the right hand side of \eqref{eq:splitting},  $(x_0,k_0;x_0,-k_0)$ is a direction of rapid decrease. Together with the analysis above, this will complete the proof of Proposition~\ref{prop:HadM/O}.

Recall the notation $f_k(x)=(2\pi)^{-2}f(x) e^{i k x}$, and that, after some choice of coordinate system,
\begin{align}
\label{eq:reformFT}
\widehat{h\cdot E\cdot f}(k,l)=\widehat{h\cdot E(f_l)}(k)\, .
\end{align}
Then we can show that
\begin{lem}
Let $(x_0,k_0)\in \mN$, with $B_M(x_0,k_0)$ intersecting $\mH$.  Identify $k_0$ with an element of $\bR^4$ under $\psi_+$. Let $K$ be a sufficiently small compact neighbourhood of $x_0$ covered by the +-Kruskal coordinate chart, and let $V\subset \bR^4\backslash \{0\}$ be a sufficiently small conic neighbourhood of $k_0$, such that $B_M(x,k)$ intersects $\mH$ in some compact set $\mathcal{U}$ for all $x\in K$ and all $k\in V$. Let $h\in C_0^\infty(\mH)$ be such that $h=1$ on $\mathcal{U}$. Then, there is a function $f\in C_0^\infty(M)$, with $f(x_0)=1$,  an open conic neighbourhood $V_{k_0}\subset \bR^4\backslash\{0\}$ of $k_0$, and $\forall $ $N\in \bN$, there are $C_N$, $\tilde C_N>0$ such that 
\begin{align}
\label{eq:LemEst}
\vert(1-h)E(f_k)\vert_{\mH}\vert&\leq \vert U_+\vert^{-\alpha/\kappa_+}\frac{ C_N}{1+\vert k\vert^N}\quad\forall k\in V_{k_0} \\
\vert E(f_k)\vert_{\mHc}\vert&\leq \vert V_c\vert^{-\alpha/\kappa_c} \frac{\tilde C_N}{1+\vert k\vert^N}\quad\forall k\in V_{k_0}\, ,
\end{align}
\end{lem}

\begin{proof}
Let us define the set 
\begin{align*}
B_M(K,V)=\{x^\prime\in \tilde M: x^\prime\in B_M(x,k)\text{ for some }x\in K,\, k\in V \}\, .
\end{align*}

Let $\Sigma_{k_0}$ be a Cauchy surface of $\tilde M$ such that $\Sigma_{k_0}$ coincides with $\mH$ in $\mH\cap \{U_1<U_+<U_2\}$. $U_1<0$, $U_2>U_f$ (see proof of \ref{prop:wellDef} for the definition) are chosen such that $\supp(h)\subset \Sigma_{k_0}\cap\mH$. 

Let $\Sigma_{\pm}$ be two other Cauchy surfaces such that $\Sigma_{k_0}\subset I^+(\Sigma_-)\cap I^-(\Sigma_+)$ and $K\subset J^+(\Sigma_+)\backslash \Sigma_+$. 

Let $\tilde h$, $ h^\prime\in C_0^\infty(\tilde M)$ be  real positive functions such that $\tilde h\vert_{\mH}=h$, $\supp(\tilde h)\cap \mHc=\emptyset$, $\tilde h+ h^\prime=1$ in a  neighbourhood of $J^-(K)\cap J^+(\Sigma_-)\cap J^-(\Sigma_+)$ and that there is an open neighbourhood $\mathcal{V}\subset \tilde M$ of $ B_M(K,V)$ which does not intersect $\supp( h^\prime)$.

Let $\eta\in C^\infty(\tilde M)$ be defined by $\eta=1$ in a neighbourhood of $B_M(K,V)$, such that $\supp(\eta) \subset \mathcal{V}$.

Finally, let $\chi \in C^\infty(\tilde M)$ be a cutoff-function which is equal to one in $J^+(\Sigma_+)$ and vanishes in $J^-(\Sigma_-)$.

\begin{figure}
\centering
\begin{subfigure}{0.6\textwidth}
\includegraphics[width=\textwidth]{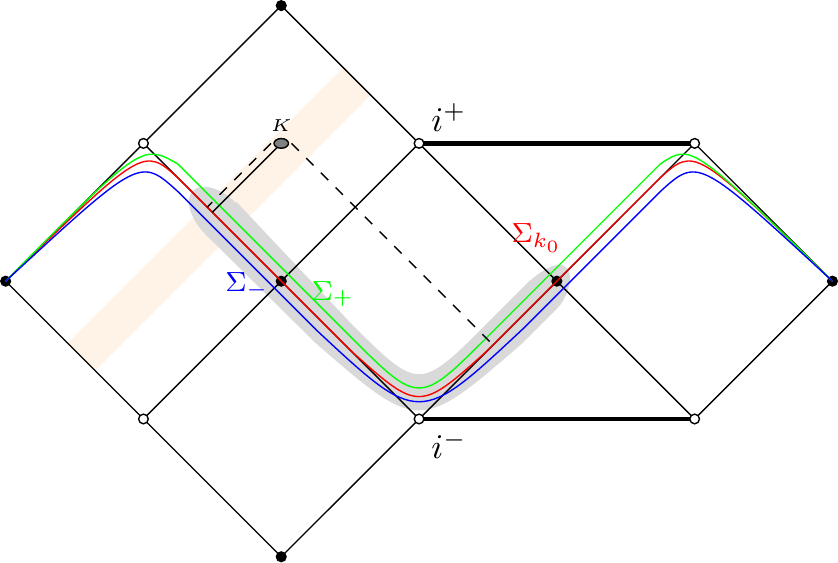}
\end{subfigure}
\begin{subfigure}{0.3\textwidth}
\includegraphics[width=\textwidth]{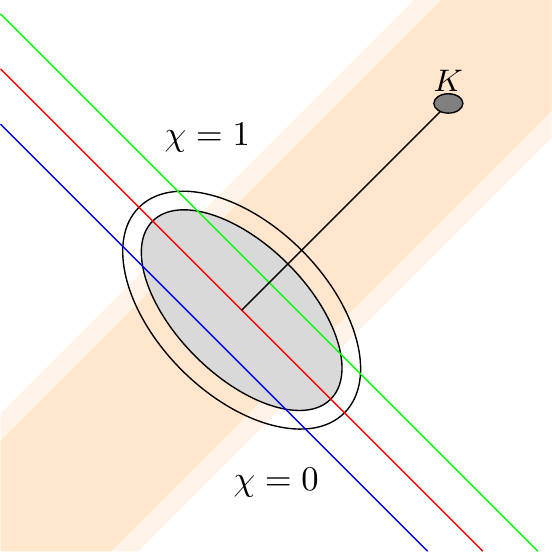}
\end{subfigure}
\caption{Left: The three Cauchy surfaces are, from top to bottom, $\Sigma_+$, $\Sigma_{k_0}$ and $\Sigma_-$. The small, dark gray region is $K$. The line joining $K$ and $\mH$ indicates the bicharacteristic $B_M(x_0,k_0)$. The dashed lines mark $J^-(K)$. The light stripe around $B_M(x_0,k_0)$ shows the neighbourhood $\mathcal{V}$ of $B_M(K,V)$, on which $h^\prime=0$. The light gray region  around $J^-(K)\cap J^-(\Sigma_+)\cap J^+(\Sigma_-)$ indicates $\supp(\tilde h+h^\prime)$.
 Right: The two ellipses indicate $\tilde h=1$ (inner, shaded ellipse), and $\supp(\tilde h)$. The shaded stripe indicates $\eta=1$ (darker shade) and $\supp(\eta)$(lighter shade). The function $\chi$ is equal to one above the topmost diagonal line, corresponding to $\Sigma_+$ and vanishes below the bottommost one, corresponding to $\Sigma_-$.}
\label{Fig:Construction}
\end{figure}

An illustration is shown in Fig. \ref{Fig:Construction}.

Then, we note that for any function $g\in C_0^\infty(M)$ with support in $K$, $\tilde{g}\equiv\mathcal{K}(\chi E(g))\in C_0^\infty(\tilde M)$ has support contained in $J^-(K)\cap J^+(\Sigma_-)\cap J^-(\Sigma_+)$ and $E(\tilde{g})=E(g)$. In addition, we have $\tilde{g}=\tilde h \tilde{g}+ h^\prime \tilde{g}$ by construction.

Applying the linearity of $E$, as well as the properties of the fundamental solutions, we thus find for any such function
\begin{align}
\label{eq:rewriteEfk}
E(g)=&E(\tilde{g})=E(\tilde h \tilde{g})+E(h^\prime \tilde{g})\\\nonumber
=&E^+(\tilde h\mathcal{K}(\chi E(g)))+E^-(\tilde h \mathcal{K}((1-\chi)E(g)))+E(h^\prime \tilde{g})\\\nonumber
=&E^+(\mathcal{K}(\tilde h\chi E(g)))+E^-(\mathcal{K}(\tilde h (1-\chi)E(g)))-E^+([\mathcal{K},\tilde h]\chi E(g))\\\nonumber 
&-E^-([\mathcal{K},\tilde h](1-\chi)E(g))+E(h^\prime \tilde{g})\\\nonumber
=&\tilde h E(g)  -E^+([\mathcal{K},\tilde h]\chi E(g))-E^-([\mathcal{K},\tilde h](1-\chi)E(g))+E(h^\prime \tilde{g})\, .
\end{align}
Since $\tilde h\vert_{\mH}=h$ and $\supp(\tilde h)\cap \mHc=\emptyset$, $(1-h)E(g)\vert_{\mH}$ and $E(g)\vert_{\mHc}$ are determined by the last three terms in the last line above restricted to $\mH$ or $\mHc$ respectively.

By a careful consideration of the supports of the different functions, we find that the second and third term above satisfy
\begin{align*}
\supp([\mathcal{K},\tilde h]\chi E(g))\cap\mathcal{V} \subset J^+(\Sigma_+)\\
\supp([\mathcal{K},\tilde h](1-\chi)E(g))\cap \mathcal{V}\subset J^-(\Sigma_-)\, .
\end{align*}

This allows us to further split them as
\begin{align*}
E^+([\mathcal{K},\tilde h]\chi E(g))&=E^+(\eta[\mathcal{K},\tilde h]\chi E(g))+E^+((1-\eta)[\mathcal{K},\tilde h]\chi E(g))\, .\\
E^-([\mathcal{K},\tilde h](1-\chi) E(g))&=E^-(\eta[\mathcal{K},\tilde h](1-\chi) E(g))+E^-((1-\eta)[\mathcal{K},\tilde h](1-\chi) E(g))\, .
\end{align*}
$\eta[\mathcal{K},\tilde h]\chi E(g)$ is then supported in $J^+(\Sigma_+)\cap J^+(\mH\cup\mHc)$, while $\eta[\mathcal{K},\tilde h](1-\chi)E(g)$ is supported in $J^-(\Sigma_-)\cap J^-(\mH\cup\mHc)$. Hence, the corresponding pieces will not give any contribution on $\mH\cup \mHc$.
 
After setting up this construction, the next step is to find the compactly supported function $f$. We will do so by considering the remaining terms that we have identified above.

Let us start with the last term, $E(h^\prime \tilde{g})$.
We note that
\begin{align*}
h^\prime \tilde{g}=h^\prime \Box_g \chi E(g)+2 h^\prime \nabla_a\chi\nabla^aE(g)\, .
\end{align*}
Applying the property \eqref{eq:WF(E)} of the commutator function, the support properties of $h^\prime$, and the fact that differentiation and multiplication by smooth functions do not increase the wavefront set, we find that
\begin{subequations}
\begin{align}
\label{eq:WFS1}
(y,l;x_0, k_0)\notin \WF( (h^\prime \Box_g\chi) \cdot E)\\
(y,l;x_0, k_0)\notin \WF((h^\prime \nabla_a \chi) \cdot \nabla^a E)
\end{align}
\end{subequations} 
$\forall (y,l)\in T^*M$, using the identification of $(x_0,k_0)$ with an element of $\bR^4\times \bR^4$ under $\psi_+$. Let us also fix some coordinate system for $y$ and $l$ which covers $\supp(h^\prime)$. Then, by Lemma \ref{lem:rapDec}, there is a function $f_1\in C_0^\infty(M)$ with $f_1 (x_0)=1$ and $\supp( f_1)\subset K$ and an open conic neighbourhood $\tilde{V}_{ k_0}\subset\bR^4\backslash \{0\}$ of $k_0$, and for any $N,N^\prime\in \bN$ there is a constant $\tilde{C}_{NN^\prime}>0$ such that 
\begin{subequations}
\label{eq:est1}
\begin{align}
\vert\widehat{(h^{\prime}\Box_g\chi \otimes f_1) \cdot E}\vert(l,k)\leq \frac{\tilde{C}_{NN^\prime}}{(1+\vert l\vert^{N^\prime})(1+\vert k\vert^{N})}\quad \forall (l,k)\in \bR^4\times \tilde{V}_{ k_0}\, ,\\
\vert\widehat{(h^{\prime} \nabla_a \chi \otimes f_1) \cdot \nabla^a E}\vert(l,k)\leq \frac{\tilde{C}_{NN^\prime}}{(1+\vert l\vert^{N^\prime})(1+\vert k\vert^{N})}\quad\forall (l,k)\in\bR^4\times \tilde{V}_{k_0}\,. 
\end{align}
\end{subequations}

We now turn to the remaining pieces of the second and third term in the last line of \eqref{eq:rewriteEfk}.
The support of $(1-\eta)[\mathcal{K},\tilde h]\chi$ is compact and disjoined from $ B_M(K,V)$. Thus, this term can be handled in the same way as $ h^\prime \tilde{g}$ by using Lemma \ref{lem:rapDec}. We find some open conic neighbourhood $V^\prime_{k_0}$ of $ k_0$ in the +-Kruskal coordinates and some function $f_2 \in C_0^\infty(M)$ supported in $K$ with $f_2(x_0)=1$ such that an estimate of the form \eqref{eq:est1} with some $C^\prime_{NN^\prime}>0$ holds for $((1-\eta)[\mathcal{K},\tilde h]\chi \otimes  f_2) \cdot E$ for all  covectors $(l,k)\in \bR^4\times {V}^\prime_{k_0}$ for any $N,\, N^\prime \in \bN$.

The term $E^-((1-\eta)[\mathcal{K},\tilde h](1-\chi)E(g))$ can be treated by an application of Lemma \ref{lem:rapDec} in the same way to get $f_3\in C_0^\infty(M)$, $V^{\prime\prime}_{k_0}\subset \bR^4\backslash \{0\}$ and constants $C^{\prime\prime}_{NN^\prime}$ such that an estimate of the form \eqref{eq:est1} holds for $((1-\eta)[\mathcal{K},\tilde h](1-\chi) \otimes  f_3) \cdot E$ for all $N,\, N^\prime\in \bN$ for all $(l,k)\in \bR^4\times {V}^{\prime\prime}_{k_0}$.

By an application of \cite[Lemma 8.1.1]{Hoermander}, the above estimates continue to hold if we replace $f_1$, $f_2$ and $f_3$ by
\begin{align}
f\equiv f_1 \cdot f_2 \cdot f_3\in C_0^\infty(K)\, .
\end{align}
All three estimates hold for $k\in V_{k_0}$, where we define $V_{k_0}$ to be the intersection of $\tilde{V}_{k_0}$,  $V^{\prime}_{k_0}$ and  $V^{\prime\prime}_{k_0}$.

In the following, we return to the estimate  \eqref{eq:est1} with $f_1$ replaced by $f$. 
By taking $N^\prime$ large enough and applying the Fourier inversion formula and \eqref{eq:reformFT}, one can conclude from  \eqref{eq:est1} that for any $N\in\bN$, there is a positive constant $C_N$ so that
\begin{align*}
\norm{ h^\prime \tilde{f_k}}_{C^m}\lesssim \frac{C_N}{1+\vert k\vert^N}\quad \forall k\in V_{ k_0} 
\end{align*}
and therefore with the estimates \eqref{eq:estU} and \eqref{eq:estV} from \cite{Hintz:2015}
\begin{subequations}
\begin{align}
\vert E( h^\prime \tilde{f_k})\vert_{\mH}\vert \lesssim C \vert U_+\vert^{-\alpha/\kappa_+}\frac{C_N}{1+\vert k\vert^N}\quad \forall k\in V_{k_0}\\
\vert E( h^\prime \tilde{f_k})\vert_{\mHc}\vert \lesssim C \vert V_c\vert^{-\alpha/\kappa_c}\frac{C_N}{1+\vert k\vert^N} \forall k\in V_{ k_0}
\end{align}
\end{subequations}
for any $N\in \bN$ for some positive constants $C$,$C_N$ .

Similar estimates can be obtained for the other two terms in the same way. One finds
\begin{subequations}
\begin{align}
\vert E^+((1-\eta)[\mathcal{K},\tilde h]\chi E(f_k))\vert_{\mH}\vert\lesssim C \vert U_+\vert^{-\alpha/\kappa_+}\frac{C_N}{1+\vert k\vert^N}\quad \forall k\in {V}_{ k_0}\\
\vert E^+((1-\eta)[\mathcal{K},\tilde h]\chi E(f_k))\vert_{\mHc}\vert\lesssim C \vert V_c\vert^{-\alpha/\kappa_c}\frac{C_N}{1+\vert k\vert^N}\quad \forall k\in {V}_{ k_0}
\end{align}
\end{subequations}
and
\begin{subequations}
\begin{align}
\vert E^-((1-\eta)[\mathcal{K},\tilde h](1-\chi) E(f_k))\vert_{\mH}\vert\lesssim C \vert U_+\vert^{-\alpha/\kappa_+}\frac{C_N}{1+\vert k\vert^N}\quad \forall k\in {V}_{ k_0}\\
\vert E^-((1-\eta)[\mathcal{K},\tilde h](1-\chi) E(f_k))\vert_{\mHc}\vert\lesssim C \vert V_c\vert^{-\alpha/\kappa_c}\frac{C_N}{1+\vert k\vert^N}\quad \forall k\in {V}_{ k_0}
\end{align}
\end{subequations}
for some $C$, $C_N$ for any $N$. 

Adding up the different pieces then finishes the proof of the lemma.
\end{proof}

Let us return to \eqref{eq:splitting}, and consider for example the second term. Let us multiply the term by $f\otimes f$, where $f$ is the function from the above Lemma. Working in +-Kruskal coordinates, the Fourier transform of this product, evaluated at $(k^\prime,k)$, can be written as
\begin{align*}
A_+\left((1-h)\cdot E(f_{k^\prime})\vert_{\mH},h\cdot E(f_k)\vert_{\mH}\right)
\end{align*}

From the above lemma, we know that $\vert(1-h)\cdot E(f_{k^\prime})\vert_{\mH}\vert$ is rapidly decreasing for $k^\prime$ in a neighbourhood of $k_0$.
It only remains to note, using the estimates \eqref{eq:estU} and \eqref{eq:estV}, that $\vert h E(f_k)\vert_{\mH}\vert\vert\lesssim \vert U_+\vert^{-\alpha/\kappa_+} C (1-\vert k\vert^M)$ for some fixed $M$, i.e. they grow at most polynomially in $k$. The polynomial growth is suppressed by the rapid decay of the other part in the conic neighbourhood
\begin{align*}
\{(l,k)\in \bR^8\backslash \{0\}: 1/2 \vert l\vert<\vert k\vert<2\vert l\vert, l\in V_{k_0}\}
\end{align*}
of $(k_0,-k_0)$ \cite{Dappiaggi:2009}. Combining this with the estimates from the proof of \ref{prop:wellDef}, we find that $(x_0,k_0;x_0,-k_0)$ is indeed a direction of rapid decrease for this term. The argument for the other terms  works along the same lines. This shows $(x_0,k_0;x_0,-k_0)$ is a direction of rapid decrease for the remaining pieces of $w_+$ in \eqref{eq:splitting} and for $w_c$.

Together with the analysis of the first piece in \eqref{eq:splitting}, this shows that $(x_0,k_0)$ is in $\mN^+$ if $(x_0,k_0;x_0,k_0)$ is in $\WF^\prime(w)$.
\end{proof}

Thus, also taking into account our previous results from Proposition~\ref{prop:HadO} and Lemma~\ref{lem:Diag}, the state determined by $w$ has the Hadamard property. 

\section{Summary}
\label{sec:sum}
In this paper, we have constructed the Unruh state for a free real scalar field on a Kerr-de Sitter spacetime.

For technical reasons, we had to restrict ourselves to either slow rotation, i.e. small $a$, and moderate cosmological constant $\Lambda$, or to small $\Lambda$ and at most moderate $a$ to show the well-definedness of our two-point function in Proposition~\ref{prop:wellDef}. The condition of having either $a$ or $\Lambda$ small could be dropped once mode stability results become available for the whole parameter range of sub-extremal Kerr-de Sitter black holes. Those results are believed to hold, but are difficult to show rigorously. The condition that both $a$ and $\Lambda$ should be at most moderately large however is directly connected to our proof of the Hadamard property of the Unruh state. In particular, it is necessary for the validity of Lemma~\ref{lem:region}, which guarantees that all null geodesics not ending at one of the horizons in the past must cross a region in which the vector fields $\partial_{t_+}$ and $\partial_{t_c}$ are both time-like. Lifting this restriction would thus require a new strategy for the proof.

We have defined the two-point function for our state using the Kay-Wald two-point function \cite{Kay:1988} on the horizons $\mH$ and $\mHc$. Making use of the decay results from \cite{Hintz:2015}, it was shown that the two-point function is well-defined, and can indeed be considered as the two-point function of a quasi-free Hadamard state on the CCR- algebra of the free scalar field on the Kerr-de Sitter spacetime.

This is not a contradiction to the no-go theorem of Kay and Wald \cite{Kay:1988}, since we expect that the Hadamard property of the state will break down at $\mH$ and $\mHc$, see also \cite[Rem. 8.4]{Gerard:2020}.

We have also seen in Lemma~\ref{lem:KMSlike} that when restricted to real testfunctions with support in the exterior region $\rI$, the Unruh state is "KMS-like" \cite{Dappiaggi:2009}. Roughly speaking, this means that asymptotically near $\mH$, the state is thermal with inverse temperature $2\pi\kappa_+^{-1}$ with respect to the isometries generated by $\partial_{t_+}$, while asymptotically near $\mHc$, it is  thermal with inverse temperature $2\pi\kappa_c^{-1}$ with respect to the isometries generated by $\partial_{t_c}$. Or, stated differently, in the asymptotic past, "in"-movers and "out"-movers are thermally populated at different temperatures. This behaviour is exactly what one would expect from the generalization of the Schwarzschild Unruh vacuum to Kerr-de Sitter.

Moreover, the form of the two-point function derived in Proposition~\ref{prop:K_X} indicates that the quantum field in this state is expanded in positive-frequency modes with respect to the coordinate $U_+$ outgoing from the past event horizon and modes with positive frequency with respect to $V_c$ incoming from the past cosmological horizon. Therefore, the Unruh state constructed in this paper agrees with the one used for the numerical investigation of the evaporation of rotating black holes in \cite{Gregory:2021}.

Considering also the physical motivation for the Unruh state on Schwarzschild, the Unruh state on Kerr-de Sitter is a physically well-motivated state. Its rigorous construction presented here is thus an important step for the study of quantum effects on rotating black hole spacetimes.

\begin{acknowledgments}
{\bf Acknowledgements:} I would like to thank S. Hollands for suggesting this topic. I would also like to thank him and J. Zahn for fruitful discussions. This work has been funded by the Deutsche Forschungsgemeinschaft (DFG) under the Grant No. 406116891 within the Research Training Group RTG 2522/1.
\end{acknowledgments}

\appendix

\section{A technical Lemma}\label{sec:A1}

In this appendix, we prove a technical Lemma that is used in the proof of the Hadamard property of our state. In particular, consider a statement on the wavefront set such as \eqref{eq:WFS1},
\begin{align*}
(x,l;y_0,k_0)\notin \WF(h\cdot D)\quad \forall (x,l)\in X\times \bR^n\,  
\end{align*}
for some $h\in C_0^\infty(X)$, $D\in\mathcal{D}^\prime(X\times Y)$ and $X,Y\subset \bR^n$.
 Then, according to Def. \ref{def.:WFS}, for any $(x,l)$, there exists a function $\Phi_{(x,l)}\in C_0^\infty(X\times Y)$ with $\Phi_{(x,l)}(x,y_0)=1$ and an open conic neighbourhood $V_{(x,l)}\subset (\bR^n\times\bR^n)\backslash\{0\}$ of $(l,k_0)$ such that for any $N\in \bN$ there is a positive constant $C^{(x,l)}_N>0$ with 
\begin{align*}
\vert\widehat{\Phi_{(x,l)} h\cdot D}\vert(l^\prime,k^\prime)\leq \frac{C^{(x,l)}_N}{(1+\vert(l^\prime,k^\prime)\vert)^N}\quad \forall (l^\prime,k^\prime)\in V_{(x,l)}\, .
\end{align*} 
The Lemma below shows that under an additional assumption, we can combine the estimates for each individual covector $(x,l)$ to one estimate holding in a neighbourhood of all $l$ and all $x\in \supp(h)$. In addition, for this estimate we can choose the  compactly supported function $\Phi$ to be of the form $\chi(x)f(y)$, with $\chi(x)=1$ on the support of $h$ and $f\in C_0^\infty(Y)$ can be chosen such that its support is contained in any arbitrary but fixed compact neighbourhood of $y_0$.
\begin{lem}
\label{lem:rapDec}
Let $X,\,Y$ $\subset \bR^n$.
Let $(y_0,k_0)\in Y\times(\bR^n\backslash \{0\})$, and let $K$ be any compact neighbourhood of $y_0$. Let $D\in \mathcal{D}^\prime(X\times Y)$ such that $(x,k;y,0)\notin \WF(D)$ for all $x\in X$, $y\in Y$, $k\in \bR^n\backslash \{0\}$. Let $h\in C_0^\infty(X)$ such that 
\begin{align*}
(x,l;y_0,k_0)\notin \WF(h\cdot D)\quad \forall (x,l)\in X\times \bR^n\,  .
\end{align*}
Then we can find a function $f\in C_0^\infty(Y)$ with $f(y_0)=1$ and support in $K$, and an open conic neighbourhood $V_{k_0}\subset \bR^n\backslash\{0\}$ of $k_0$ so that for any $N,N^\prime\in \bN$ there are positive constants $C_{NN^\prime}$ satisfying, 
\begin{align*}
\vert\widehat{ (h\otimes f)\cdot D}\vert(l,k)\leq \frac{C_{NN^\prime}}{(1+\vert l\vert^N)(1+\vert k\vert^{N^\prime})}\quad \forall l\in \bR^4\, , \;  k\in V_{k_0}\, .
\end{align*}
\end{lem}

\begin{proof}
One key ingredient to this proof is \cite[Lemma 8.1.1]{Hoermander}:  Let $v \in \mathcal{E}^\prime(Z)$, $Z\subset \bR^m$, and $\phi\in C_0^\infty(Z)$. Then if $(x,k)\in Z\times\bR^m\backslash \{0\}$ is a direction of rapid decrease for $v$, it is a direction of rapid decrease for $\phi\cdot v$.

By the definition of the wavefront set and our assumptions, for any $(x,l)\in \supp(h)\times \bR^n$, there exists a function $\Phi_{(x,l)}\in C_0^\infty(X\times Y)$ with $\Phi_{(x,l)}(x,y_0)=1$ and an open conic neighbourhood $V_{(x,l)}\subset (\bR^n\times\bR^n)\backslash\{0\}$ of $(l,k_0)$ such that for any $N\in \bN$ there is a positive constant $C^{(x,l)}_N>0$ with 
\begin{align*}
\vert\widehat{\Phi_{(x,l)} h\cdot D}\vert(l^\prime,k^\prime)\leq \frac{C^{(x,l)}_N}{(1+\vert(l^\prime,k^\prime)\vert)^N}\quad \forall (l^\prime,k^\prime)\in V_{(x,l)}\, .
\end{align*} 

We can also assume that $\Phi_{(x,l)}\geq 0$. Otherwise, we could by \cite[Lemma 8.1.1]{Hoermander} multiply with another $C_0^\infty$-function  $\chi$ with $\chi(x,y_0)=1$, such that $\chi \Phi_{(x,l)}\geq 0$.

Instead of labelling $V_{(x,l)}$ and $\Phi_{(x,l)}$ by $l$, we can equally well label them by  $\lambda=l/\vert(l,k_0)\vert$. The new label $\lambda$ lies in the open ball of unit radius around the origin in $\bR^n$. So far, this is only a relabelling, which is better suited for the following argument.

Since the sets $V_{(x,\lambda)}$ are conic, we will as a simplification only consider their projection to $\bS^{2n-1}=\{(l^\prime,k^\prime):\vert(l^\prime,k^\prime)\vert=1\}$. The projection of $V_{(x,\lambda)}\ $ to $\bS^{2n-1}$ is an open neighbourhood of $(\lambda,\sqrt{1-\vert\lambda\vert^2}/\vert k_0\vert k_0)$.

By assumption, we know that $(x,l;y,0)\notin \WF(h\cdot D)$. Hence, we find open conic neighbourhoods and compactly supported functions as above for $\vert\lambda\vert=1$. Thus, for all $x$, we now have functions $\Phi_{(x,\lambda)}$ and conic sets $V_{(x,\lambda)}$ for all $\lambda$ in the closed unit ball around the origin in $\bR^n$. The projections of the sets $V_{(x,\lambda)}$ to $\bS^{2n-1}$ then form an open cover of the compact set
\begin{align*}
 \{(\lambda^\prime, \sqrt{1-\vert\lambda^\prime\vert^2}/\vert k_0\vert k_0)\in \bR^n\times\bR^n:\vert\lambda^\prime\vert\leq 1\}\,.
\end{align*} 

As a result, for any $x$, the open cover of this set by $\{\bS^{2n-1} \cap V_{(x,\lambda)}\}_{\vert\lambda\vert\leq 1}$, has a finite open subcover $\{\bS^{2n-1}\cap V_{(x,\lambda_i)}\}_{i=1,\dots ,  K}$ with corresponding functions $\Phi_{(x,\lambda_i)}$.

We then define $\Phi_{x}=\prod\limits_{i=1}^K\Phi_{(x,\lambda_i)}\in C_0^\infty(X\times Y)$. By \cite[Lemma 8.1.1]{Hoermander} with $\phi=\prod\limits_{i\neq j}\Phi_{(x,\lambda_i)}$, $v=\Phi_{(x,\lambda_j)}h\cdot D$, one can show that $\forall N\in\bN$, there are constants $C^{x}_N$ with
\begin{align*}
\vert\widehat{\Phi_{x} (h\otimes f)\cdot D}\vert(l^\prime,k^\prime)\leq \frac{C^{x}_N}{(1+\vert(l^\prime,k^\prime)\vert)^N}\,\quad \forall (l^\prime,k^\prime)\in V_{(x,\lambda_j)}\, . 
\end{align*}
Varying $j$ from $1$ to $K$, this holds for all $(l^\prime,k^\prime)\in V_{x}\equiv \bigcup_i V_{(x,\lambda_i)}$, and hence for all  $(l^\prime,k^\prime)\in\bR^n\times V^{x}_{k_0}$, with $V^{x}_{k_0}$ the open conic neighbourhood of $k_0$ given by
\begin{align*}
 V^{x}_{k_0}=\{k\in \bR^n : (l,k)\in \bigcup_i V_{(x,\lambda_i)}\forall l\in \bR^n\}\,.
\end{align*}

Next, let us define 
\begin{align*}
\mathcal{U}_{x}^\epsilon=\left\{(x^\prime,y^\prime)\in X\times Y: \Phi_{x}(x^\prime,y^\prime)>\epsilon\right\}
\end{align*}
for some small $\epsilon>0$.
$\{\mathcal{U}_{x}^\epsilon\}_{x\in \supp(h)}$ forms an open cover of $\supp(h)\times \{y_0\}$. Hence, we can find a finite open subcover
$\{\mathcal{U}_{x_i}^\epsilon\}_{i=1,\dots , M}$ of $\supp(h)\times \{y_0\}$ and corresponding functions $\Phi_i=\Phi_{x_i}$ which then satisfy
\begin{align*}
\sum\limits_{i=1}^M\Phi_i(x^\prime,y^\prime)\geq \epsilon\quad \forall (x^\prime,y^\prime)\in \supp(h)\times \pi_Y\left(\bigcap\limits_{i=1}^M \mathcal{U}^\epsilon_{x_i}\right)\, ,
\end{align*}
where $\pi_Y:X\times Y\to Y$ is the projection to $Y$.
Let $\chi\in C_0^\infty (X\times Y)$, such that
\begin{align*}
\chi= \begin{cases} \frac{1}{\sum\limits_i \Phi_i}\, : & \sum\limits_{i=1}^M \Phi_i \geq \frac{\epsilon}{2} \\
0\, : &\sum\limits_{i=1}^M \Phi_i\leq \frac{\epsilon}{4}
\end{cases}\, .
\end{align*}
Let $f\in C_0^\infty(Y)$ be supported in $\pi_Y\left(\bigcap\limits_{i=1}^M \mathcal{U}^\epsilon_{x_i}\right)\cap K$ and let $f(y_0)=1$.
Then $\chi(x,y)f(y)\in C_0^\infty(X\times Y)$.

By \cite[Lemma 8.1.1]{Hoermander}, for any $i$ and for any $N\in \bN$, there are positive constants $C^{i}_N$, such that 
\begin{align*}
\vert\widehat{f\chi\Phi_{i} h\cdot D}\vert(l,k)\leq \frac{C^{i}_N}{(1+\vert(l,k)\vert)^N}\quad\forall (l,k)\in V_{x_i}\, .
\end{align*}
Hence
\begin{align*}
\vert\widehat{(h\otimes f)\cdot D}\vert(l,k)&=\vert\widehat{\sum\limits_{i=1}^Mf\chi\Phi_i h\cdot D}\vert(l,k)\\
& \leq \sum\limits_{i=1}^M \vert\widehat{f\chi\Phi_i h\cdot D}\vert(l,k)\\
& \leq \sum_{i=1}^M\frac{ C^{i}_N}{(1+\vert(l,k)\vert)^N}\\
&\leq \frac{\tilde{C}_N}{(1+\vert(l,k)\vert)^N}
\end{align*}
for all $(l,k)\in \bigcap_{i=1}^MV_{x_i}\supset \bR^n\times V_{k_0}$, with $V_{k_0}=\bigcap_{i=1}^MV^{x_i}_{k_0}$. 

It remains to note that the euclidean norm of $(l,k)$ in $\bR^{2n}$ is equivalent to $\vert l\vert+\vert k\vert$, and that by an application of the binomial formula we get for any $a,b>0$ and $N>M\geq 0$, 
\begin{align*}
(1+a+b)^N\geq 1+a^M+b^{N-M}+a^Mb^{N-M}=(1+a^M)(1+b^{N-M})\, .
\end{align*}
\end{proof}

\section{Null geodesics on Kerr-de Sitter}\label{sec:A2}

In this appendix, we collect some results on the null geodesics on Kerr-de Sitter. Most of these results can be found in \cite{Hackmann:2010, Salazar:2017, Borthwick:2018}. They extend the ones obtained in \cite{ONeill:1995} and \cite{Gerard:2020} for Kerr spacetimes to Kerr-de Sitter spacetimes, and are used in section \ref{sec:GemSet}. We will describe the behaviour of inextendible future null geodesics $\gamma$ on $\tilde M$ and focus mostly on their radial motion.

Before we start, we mention that all horizons and the axis $\{\sin\theta=0\}$ are totally geodesic submanifolds of $\tilde M$ by \cite[Thm. 1.7.12]{ONeill:1995}. Therefore, a geodesic that does not lie entirely in one of the horizons or the axis but approaches one of these submanifolds must cross it transversally if it can be extended through that submanifold. We will begin with these geodesics, and discuss the ones contained in a horizon or the axis in the end. Note that any geodesic crossing the axis must have $L=0$. In this case, the analysis in \cite[Sec. 6]{Salazar:2017} shows that the geodesic may be extended through the axis.

Let us start with geodesics $\gamma$ intersecting $\rIII$. Since $\Delta_r<0$ on $\rIII$, we find $R(r)>0$ on $\rIII$, and $R(r)\to+\infty$ as $r\to\infty$ unless $K=E=0$. If $K=E=0$, the equation for $\theta(\tau)$ demands that also $L=0$. In this case, the geodesic is completely contained in one of the horizons, as can be seen by following the analysis in \cite[Sec. 4.2]{ONeill:1995} using the results of \cite{Borthwick:2018} on the principal null directions in Kerr-de Sitter. If we exclude this case, then $\gamma\vert_{\rIII}$ will approach $r\to\infty$ in the future, taking an infinite amount of proper time to do so, compare \cite[Sec. 4]{Salazar:2017}. To the past, the geodesics approach $r=r_c$. Let us define \cite{ONeill:1995}
\begin{align}
\bP(r)=(r^2+a^2)E-aL\, ,\\
\bD(\theta)=L-Ea\sin^2\theta\, ,
\end{align}
and go to $KdS*$- or $*KdS$-coordinates. Then
\begin{subequations}
\label{eq:u(t),v(t)}
\begin{align}
\rho^2\frac{\td v}{\td\tau}&=\frac{a\chi^2\bD}{\Delta_\theta}+\frac{\chi^2(r^2+a^2)}{\Delta_r}\left[\bP \pm\frac{\sqrt{R(r)}}{\chi}\right]\,,\\
\rho^2\frac{\td u}{\td\tau}&=\frac{a\chi^2\bD}{\Delta_\theta}+\frac{\chi^2(r^2+a^2)}{\Delta_r}\left[\bP \mp\frac{\sqrt{R(r)}}{\chi}\right]\,,
\end{align}
\end{subequations}
where the upper (lower) sign is for $\td r/\td \tau>(<)0$,  \cite[Eq. (65)-(70)]{Salazar:2017}\footnote{The same singularity structure holds for the $\tau$-derivative of the azimuthal coordinates of the $KdS*$- and $*KdS$-coordinate systems.}. In $\rIII$, any future-directed geodesic has $\td r/\td\tau>0$. Hence, it depends on the sign of $\bP(r_c)$ whether the $KdS*$- or $*KdS$-coordinates remain finite as $r$ approaches $r_c$ in finite proper time \cite{Salazar:2017}: the geodesic will cross $\mHc^L$ into $\rI$ if $\bP(r_c)>0$. If $\bP(r_c)<0$, the geodesic will cross $\mHc^R$ into $\rI^\prime$, and if $\bP(r_c)=0$, which turns $r_c$ into a simple root of $R(r)$, the geodesic will cross the bifurcation sphere $\mathcal{B}_c$ in finite proper time into $\rIII^\prime$. To observe the last case, one can change to Kruskal coordinates and follow the proof of \cite[Prop. 4.4.4]{ONeill:1995}, see also \cite{Salazar:2017}.  The discussion for region $\rIII^\prime$ in $\tilde M$ is the same, but with inverted time orientation.

Next, we consider  $\gamma$ intersecting $\rII$. Here, we have $R(r)>0$ as well, unless $K=E=L=0$. In the latter case the geodesic must be contained in a horizon. $\gamma\vert_{\rII}$ approaches $r=r_-$ to the future and $r=r_+$ to the past. It will reach the horizons or bifurcation spheres in finite proper time. To the past, the geodesic will cross $\mH^R$ into $\rI$ if $\bP(r_+)>0$ and $\mH^L$ into $\rI^\prime$ if $\bP(r_+)<0$. If $\bP(r_+)=0$, $r_+$ will becomes a simple root of $R(r)$ and the geodesic will cross through $\mathcal{B}_+$ into $\rII^\prime$.

Now, let us discuss geodesics intersecting region $\rI$. Here, $R(r)$ can have two roots, of which one might be located at $r_+$ or $r_c$, or a double root. If $R(r)$ has two roots in $\rI$ then it must be negative between them. All other cases cn be excluded by the structure of $R(r)$. On $\rI$, the vector field $V=(r^2+a^2)\partial_t+a^2\partial_\varphi$ is a future-pointing timelike vector field, and hence $\bP(r)=-g(\gamma^\prime,V)>0$ for the tangent vector $\gamma^\prime$ of $\gamma\vert_{\rI}$, compare \cite{ONeill:1995, Borthwick:2018}. This, together with \eqref{eq:u(t),v(t)}, leads to the following possibilities of radial motion for $\gamma\vert_{\rI}$:
\begin{itemize}
\item $r_+\to r_c$ or $r_c\to r_+$: The geodesic crosses $\rI$ in finite proper time from $\mHc^-$ to $\mH^R$ or from $\mH^-$ to $\mHc^L$.
\item $r_+\to r_+$ or $r_c\to r_c$: The geodesic enters $\rI$ from $ \mH^-$ or $\mHc^-$, is reflected at a simple root of $R(r)$, and exits through $\mH^R$ or $\mHc^L$. This takes finite proper time.
\item $r_+\to r_0$ or $r_c\to r_0$: The geodesic enters $\rI$ from $\mH^-$ or $\mHc^-$ at finite proper time, and then asymptotically approaches the double root $r_0$ of $R(r)$, taking infinite proper time to do so.
\item $r_0\to r_+$ or $r_0\to r_c$: The geodesic exits $\rI$ through $\mH^R$ or $\mH^L_c$, and approaches the double root $r_0$ of $R(r)$ towards the past asymptotically, taking infinite proper time to do so.
\item $r_0\to r_0$: The geodesic remains at $r=r_0$ for all $\tau\in \bR$.
\end{itemize}

Finally, let us discuss geodesics contained in one of the horizons or the rotation axis. First, the geodesics contained in one of the horizons are complete and cross through the corresponding bifurcation sphere. This can be seen by introducing Kruskal-type coordinates,  \cite[Sec. 4.4.2]{Borthwick:2018}. The geodesic contained in the axis satisfy $K=L=0$. In this case $R(r)>0$, and depending on the sign of $\td r/\td\tau$, they follow either lines of constant $v$ or constant $u$ \cite[Sec. 6]{Salazar:2017}.

After this analysis, let us also mention that in the parameter regime where Lemma~\ref{lem:region} holds, for any double root $r_0$ of $R(r)$ in $\rI$ with $E\neq0$, one can check that 
\begin{align*}
\rho^2\frac{\td t}{\td\tau}(r_0,\theta)=\frac{2\chi^2 E}{\Delta_\theta\Delta_r^\prime(r_0)}\left[r_0^2(\chi r_0+3)+a^2\cos^2\theta(\chi r_0-1)\right]\, .
\end{align*}
This is non-vanishing. One can then follow the analysis in the proof of \cite[Lemma C.4 1.i)]{Gerard:2020} to show that any inextendible null geodesic in $\rI$ satisfies $\sup_\gamma(t)=\infty$ and $\inf_\gamma (t)=-\infty$.

\bibliography{bibi}
\bibliographystyle{ieeetr}
\end{document}